%% file: tpref.tex
\pdfoutput=1
\documentclass[acmsmall, screen, nonacm]{acmart}
\makeatletter
\let\@authorsaddresses\@empty
\makeatother

\usepackage{preamble}
\usepackage{iris}
\usepackage{prob}

\title{Almost-Sure Termination by Guarded Refinement}

\author{Simon Oddershede Gregersen}
\authornote{The majority of this work was carried out while the author was affiliated with Aarhus University.}
\orcid{0000-0001-6045-5232}
\affiliation{
  \institution{New York University}
  \country{USA}
}
\email{s.gregersen@nyu.edu}

\author{Alejandro Aguirre}
\orcid{0000-0001-6746-2734}
\affiliation{
  \institution{Aarhus University}
  \country{Denmark}
}
\email{alejandro@cs.au.dk}

\author{Philipp~G. Haselwarter}
\orcid{0000-0003-0198-7751}
\affiliation{
  \institution{Aarhus University}
  \country{Denmark}
}
\email{pgh@cs.au.dk}

\author{Joseph Tassarotti}
\orcid{0000-0001-5692-3347}
\affiliation{
  \institution{New York University}
  \country{USA}
}
\email{jt4767@nyu.edu}

\author{Lars Birkedal}
\orcid{0000-0003-1320-0098}
\affiliation{
  \institution{Aarhus University}
  \country{Denmark}
}
\email{birkedal@cs.au.dk}



\begin{document}

\begin{abstract}
  Almost-sure termination is an important correctness property for probabilistic programs, and a number of program logics have been developed for establishing it.
  However, these logics have mostly been developed for first-order programs written in languages with specific syntactic patterns for looping.
  In this paper, we consider almost-sure termination for higher-order probabilistic programs with general references.
  This combination of features allows for recursion and looping to be encoded through a variety of patterns.
  Therefore, rather than developing proof rules for reasoning about particular recursion patterns, we instead propose an approach based on proving \emph{refinement} between a higher-order program and a simpler probabilistic model, in such a way that the refinement preserves termination behavior.
  By proving a refinement, almost-sure termination behavior of the program can then be established by analyzing the simpler model.

  We present this approach in the form of \thereflog{}, a higher-order separation logic for proving termination-preserving refinements.
  \thereflog{} uses probabilistic couplings to carry out relational reasoning between a program and a model.
  To handle the range of recursion patterns found in higher-order programs, \thereflog{} uses guarded recursion, in particular the principle of L\"{o}b induction.
  A technical novelty is that \thereflog{} does not require the use of transfinite step indexing or other technical restrictions found in prior work on guarded recursion for termination-preservation refinement.
  We demonstrate the flexibility of this approach by proving almost-sure termination of several examples, including first-order loop constructs, a random list generator, treaps, and a sampler for Galton-Watson trees that uses higher-order store.
  All the results have been mechanized in the Coq proof assistant.
\end{abstract}

\maketitle

\input{intro}
\input{preliminaries}
\input{ref-log}
\input{soundness}
\input{asynchronous}
\input{examples}
\input{related-work}
\input{conclusion}

\begin{acks}
  This work was supported in part by
  the \grantsponsor{NSF}{National Science Foundation}{}, grant no.~\grantnum{NSF}{2225441},
  the \grantsponsor{Carlsberg Foundation}{Carlsberg Foundation}{}, grant no.~\grantnum{Carlsberg Foundation}{CF23-0791},
  a \grantsponsor{Villum}{Villum}{} Investigator grant, no. \grantnum{Villum}{25804}, Center for Basic Research in Program Verification (CPV), from the VILLUM Foundation,
  and the European Union (\grantsponsor{ERC}{ERC}{}, CHORDS, \grantnum{ERC}{101096090}).
  Views and opinions expressed are however those of the author(s) only and do not necessarily reflect those of the European Union or the European Research Council.
  Neither the European Union nor the granting authority can be held responsible for them.
\end{acks}

\bibliography{refs}

\pagebreak
\appendix

\input{appendix}

\end{document}

%% file: intro.tex
\section{Introduction}\label{sec:introduction}

Probabilistic programs are programs that draw samples from probability distributions.
They have many applications but also complex and unintuitive behaviors.
Therefore, there has been long-standing interest in formal techniques for reasoning about them.

In particular, termination of probabilistic programs is an important property for the correctness of various sampling and consensus algorithms.
This work considers the problem of \emph{almost-sure termination} (AST), that is, whether a probabilistic program terminates with probability 1.
The problem of showing AST is known to be computationally harder than termination of deterministic programs~\cite{kaminski_hardness_2019} so there is clearly no hope of completeness.
On the other hand, the problem is decidable for large classes of probabilistic processes~\cite{brazdil_analyzing_2013}, and indeed such decision procedures have been implemented in probabilistic model checkers such as PRISM~\cite{kwiatkowska_prism_2011}.
Moreover, there is a long and rich history in the mathematical literature that has studied the termination behavior of multiple families of stochastic processes~\citep{spitzer2013principles,athreya2012branching}. 

Multiple works have developed program logics to reason about termination or expected running time of probabilistic programs.
For example, the work on weakest pre-expectation calculi by~\citet{morgan_pgcl} can, in particular, be used to prove AST, and the expected runtime transformer by~\citet{kaminski_weakest_2016} provides a compositional way to reason about the running time of probabilistic programs.
Ranking Supermartingales \cite{chakarov_probabilistic_2013,fu_termination_2019} are probabilistic analogues of ranking functions that can be used to prove termination of probabilistic programs by adapting results from martingale theory to program verification.
Others \cite{McIverMKK18} develop rules to prove termination of particularly complex iteration schemes.

The works listed above have all increased the reach of termination analysis of probabilistic programs, but they are mostly limited to first-order, imperative languages and their techniques are adapted to this particular setting.
Some works \cite{KobayashiLG19, DBLP:journals/pacmpl/AvanziniBL21} have considered stateless higher-order programs and higher-order recursion schemes, but it is unclear to what extent these approaches scale to richer languages.
For instance, in a setting that includes higher-order functions and higher-order references, recursion and thus divergence can arise in multiple ways, \eg{}, by recursion through the store, so rules that have been tailored to syntactic while loops or particular recursion schemes will be difficult to generalize and apply.

As an example, consider the randomized function ${\sf walk}$ shown below, which uses a fixed-point combinator {\sf fix} defined using Landin's knot~\cite{DBLP:journals/cj/Landin64}:
\begin{align*}
  {\sf fix} &\triangleq \Lam f . \Let r = \Alloc~(\Lam x. x) in r \gets (\Lam x. f~(\deref{r})~x) ;~\deref{r} \\
  {\sf F} &\triangleq
            \Lam f . \Lam n . \hspace{-0.5em}
            \begin{array}[t]{l}
              \If {n == 0} then {()}\\
              \Else{\If{\Flip}then{f~(n-1)}\Else{f~(n+1)}}
            \end{array} \\
  {\sf walk} &\triangleq {\sf fix}~{\sf F}
\end{align*}
The $\Flip$ expression in {\sf F} reduces uniformly at random to either $\True$ or $\False$.

The execution of ${\sf walk}~n$ depends on a complex interaction between probabilistic choice and higher-order store.
Nevertheless, ${\sf walk}~n$ almost-surely terminates for any starting value~$n$.
The reason is that the value of the argument on each successive call to $F$ follows a symmetric random walk on the natural numbers which terminates upon reaching $0$.
This is a well-known stochastic process that almost-surely terminates.
The random walk can be represented by the probabilistic transition system depicted below
\begin{center}
  \tikzset{every loop/.style={min distance=5mm,in=180,out=140,looseness=2}}
  \begin{tikzpicture}[node distance=1.5cm,on grid, auto]
    \node[state, accepting] (s0) {$0$};
    \node[state] (s1) [right=of s0] {$1$};
    \node[state] (s2) [right=of s1] {$2$};
    \node[state] (s3) [right=of s2] {$3$};
    \node (dots) [right=of s3] {$\cdots$};
    \path[->] (s1) edge [bend left] node {$\frac{1}{2}$} (s2);
    \path[->] (s2) edge [bend left] node {$\frac{1}{2}$} (s3);
    \path[->] (s3) edge [bend left] node[xshift=-7pt] {$\frac{1}{2}$} (dots);
    \path[->] (dots) edge [bend left] node[xshift=-7pt] {$\frac{1}{2}$} (s3);
    \path[->] (s3) edge [bend left] node {$\frac{1}{2}$} (s2);
    \path[->] (s2) edge [bend left] node {$\frac{1}{2}$} (s1);
    \path[->] (s1) edge [bend left] node {$\frac{1}{2}$} (s0);
  \end{tikzpicture}
\end{center}
where, in this analogy, the labels on the states represent the current value of the argument $n$, and $0$ represents a terminal state, corresponding to the fact that recursion stops when $n$ is $0$.
Since execution of ${\sf walk}~n$ terminates whenever the transition system terminates, and the transition system almost-surely terminates, ${\sf walk}~n$ must almost-surely terminate.
  
In this work, we explore a novel approach to proving AST that allows us to capture this kind of argument in a precise and formal way.
We develop \thereflog, a \emph{higher-order guarded separation logic} that allows one to prove termination of examples like ${\sf walk}~n$.
Instead of reasoning about the termination probability directly, \thereflog{} establishes a termination-preserving \emph{refinement} between a user-chosen probabilistic model (like the random walk above) and a probabilistic program (such as ${\sf walk}~n$).
Termination preservation here implies that the probability of termination of the program is at least as high as the probability of termination of the model.
Thus, if the model almost-surely terminates, so does the program.
This allows us to transfer the problem of proving AST of a probabilistic program to that of proving AST of a model, and in turn it allows us to make use of a wide set of tools for showing termination of the model.
The benefit of this approach is that it makes it possible to apply the rich theory and extensive prior work that has been developed for AST of first-order programs, without the need to adapt that work to the setting of a higher-order language.
In particular, for the motivating example at hand, we have used \thereflog{} to show that ${\sf walk}~n$ refines the symmetric random walk model and, as a consequence, that it almost-surely terminates.

To support reasoning about the different forms of recursion present in higher-order languages with higher-order store, \thereflog{} uses \emph{guarded recursion} based on step indexing \cite{nakano-later,appel-later,DBLP:journals/corr/abs-1208-3596}.
Several recent studies have used guarded recursion for termination-preserving refinement for \emph{non}-probabilistic programs \cite{DBLP:conf/esop/TassarottiJ017, transfinite-iris, trillium}, but it is not \emph{a priori} clear that their formulation of refinement can be adapted to the probabilistic setting.
In fact, to preserve termination, that prior work has had to impose various restrictions on non-determinism of programs, or to replace standard step indexing with \emph{transfinite} step indexing, and so one might expect that probabilistic termination preservation would require similar technical changes.
Surprisingly, as we show in \cref{sec:soundness}, it turns out that these workarounds are not needed in the probabilistic setting.

One limitation of the refinement approach is that it requires coming up with suitable models of programs.
If the model is very close to the original program, the refinement may be easy to show, yet analyzing the termination of the model is then no simpler.
On the other hand, if the model is considerably simpler than the program, one might worry that the intended refinement is difficult to prove.
To demonstrate empirically that \thereflog{} is effective, we have verified a range of examples, several of which demonstrate intricate use of higher-order functions and higher-order state, putting them beyond the scope of previous techniques.

\paragraph{Contributions.}
In summary,
\begin{enumerate}
\item We develop a compositional, higher-order guarded separation logic for showing termination-preserving refinement of higher-order probabilistic programs, which we use as a technique for showing almost-sure termination.
\item We identify two new and orthogonal uses cases for \emph{asynchronous coupling} \cite{clutch} as a mechanism for (a) coupling \emph{one} model step to \emph{multiple} program samplings, and (b) managing guarded recursion, that is, to eliminate later modalities \emph{now}, which could otherwise only have been eliminated \emph{in the future} if ordinary couplings had been used.
\item We demonstrate our approach on a rich set of examples, showcasing the potential both on the ``classic'' first-order examples but also on more involved implementations which make use of local and dynamically-allocated higher-order state.
  Our examples also demonstrate how our approach allows for concise, composable, and higher-order specifications that resemble specifications in non-probabilistic separation logics. 
\item All of the results presented in this paper are mechanized \cite{caliper_artifact} in the Coq proof assistant \cite{coq}, including the semantics, the logic, the mathematical analysis results, and the case studies, with the help of the Coquelicot library for real analysis \cite{coquelicot} and the Iris separation logic framework \cite{irisjournal}.
\end{enumerate}


%% file: preliminaries.tex
\section{Background and Preliminaries}\label{sec:preliminaries}
In this section, we recall some basic definitions from probability theory and we define what it means to execute a probabilistic program. 
Although the \thereflog{} approach is language generic, in this paper we fix a probabilistic ML-like language, whose semantics we describe here.
Finally, we define a notion of probabilistic coupling that will be central to the soundness of our approach.
  
\subsection{Probabilistic Semantics}

To account for non-terminating behavior, we make use of (discrete) probability \emph{sub}-distributions.
\begin{definition}
  A \emph{sub-distribution} over a countable set $A$ is a function $\mu : A \to [0,1]$ such that $\sum_{a \in A} \mu(a) \leq 1$.
  We write $\Distr{A}$ for the set of all sub-distributions over $A$.
\end{definition}

\begin{definition}
  The \emph{support} of $\mu \in \Distr{A}$ is the set of elements $\supp(\mu) \eqdef{} \set{ a \in A \mid \mu(a) > 0 }$.
\end{definition}

\begin{lemma}
  Let $\mu \in \Distr{A}$, $a \in A$, and $f : A \to \Distr{B}$.
  Then
  \begin{enumerate}
  \item $\mbind(f,\mu)(b) \eqdef{} \sum_{a \in A} \mu(a) \cdot f(a)(b)$
  \item $\mret(a)(a') \eqdef{}
    \begin{cases}
      1 & \text{if } a = a' \\
      0 & \text{otherwise}
    \end{cases}$
  \end{enumerate}
  gives monadic structure to $\DDistr$.
  We write $\mu \mbindi f$ for $\mbind(f, \mu)$.
\end{lemma}

We obtain a small-step operational semantics for both programs and models by considering them as (discrete-time) \emph{Markov chains}.

\begin{definition} \label{def:markov-chain}
  A (sub)-\emph{Markov chain} over a countable set $\Markov$ is a function $\stepdistr \colon \Markov \to \Distr{\Markov}$.
\end{definition}
Given a Markov chain over $\Markov$ and a decidable predicate $\tofinal : \Markov \to \mProp$ such that if $\tofinal(\mstate)$ then $\stepdistr(\mstate)(\mstate') = 0$ for all $\mstate'$, we define what it means for a Markov chain to evaluate to a final state.
First, we define a stratified execution distribution $\exec_{n}\colon \Markov \to \Distr{\Markov}$ by induction on $n$:
\begin{align*}
  \exec_{n}(\mstate) \eqdef{}
  \begin{cases}
    \mathbf{0} & \text{if}~ \neg \tofinal(\mstate)~\text{and}~n = 0 \\
    \mret(\mstate) & \text{if}~\tofinal(\mstate) \\
    \stepdistr(\mstate) \mbindi \exec_{(n - 1)} & \text{otherwise}
  \end{cases}
\end{align*}
where $\mathbf{0}$ denotes the everywhere-zero sub-distribution.
Observe that the value $\exec_{n}(\mstate)(\mstate')$ denotes the probability of stepping from a state $\mstate$ to a final state $\mstate'$ in at most $n$ steps.
The probability that an execution starting from a state $\mstate$ reaches a final state $\mstate'$ is the limit of its stratified approximations, which exists by monotonicity and boundedness:
\begin{align*}
  \exec(\mstate)(\mstate') \eqdef{} \lim_{n \to \infty} \exec_{n}(\mstate)(\mstate').
\end{align*}
The probability that an execution from state $\mstate$ terminates is thus $\execTerm(\mstate) \eqdef{} \sum_{\mstate' \in \Markov} \exec(\mstate)(\mstate')$.
By the monotone convergence theorem we get the lemma below, which intuitively says that it suffices to consider all finite approximations to bound the termination probability.
\begin{lemma}\label{lem:exec-mct}
  If $\sum_{\mstate' \in \Markov} \exec_{n}(\mstate)(\mstate') \leq r$ for all $n$ then $\execTerm(\mstate) \leq r$.
\end{lemma}

\paragraph{\thelang.}
The syntax of $\thelang{}$, the programming language we consider throughout this paper, is defined by the grammar below.
\begin{align*}
  \val, \valB \in \Val \bnfdef{}
  & z \in \integer \ALT
    b \in \bool \ALT
  \TT \ALT
  \loc \in \Loc \ALT
    \Rec \lvarF \lvar = \expr \ALT
  (\val,\valB) \ALT
  \Inl \val  \ALT
    \Inr \val \ALT
  \\
  \expr \in \Expr \bnfdef{}&
  \val \ALT
  \lvar \ALT
    \expr_1 + \expr_2 \ALT
    \expr_1 - \expr_2 \ALT
    \ldots \ALT
  \expr_1~\expr_2 \ALT
  \If \expr then \expr_1 \Else \expr_2 \ALT
  \Fst \expr \ALT
  \Snd \expr \ALT \\
  & \Match \expr with \Inl \val~ => \expr_1 | \Inr \valB => \expr_2 end \ALT
  \Alloc \expr \ALT
  \deref \expr \ALT
  \expr_1 \gets \expr_2 \ALT
  \Rand \expr
  \\
  \lctx \in \Ectx \bnfdef{} &
 -
 \ALT \expr\,\lctx
 \ALT \lctx\,\val
 \ALT \Alloc\lctx
 \ALT \deref \lctx
 \ALT \expr \gets \lctx
 \ALT \lctx \gets \val
 \ALT \Rand\lctx
\ALT \ldots
\\
  \state \in \State \eqdef{}& \Loc \fpfn \Val
  \\
  \cfg \in \Conf \eqdef{}& \Expr \times \State
\end{align*}
The term language is mostly standard:
$\Alloc \expr$ allocates a new reference, $\deref \expr$ dereferences the location $\expr$ evaluates to, and $\expr_{1} \gets \expr_{2}$ assigns the result of evaluating $\expr_{2}$ to the location that $\expr_{1}$ evaluates to.
We introduce syntactic sugar for lambda abstractions $\Lam \var . \expr$ defined as $\Rec {\_} \var = \expr$,
let-bindings $\Let \var = \expr_{1} in \expr_{2}$ defined as $(\Lam \var . \expr_{2})~\expr_{1}$, and sequencing $\expr_{1} ; \expr_{2}$ defined as $\Let \_ = \expr_{1} in \expr_{2}$.

The language has a call-by-value Markov chain semantics $\stepdistr : \Conf \to \Distr{\Conf}$ defined using evaluation contexts $\lctx \in \Ectx$.
We set $\tofinal(\expr, \sigma) \eqdef{} \left(\expr \in \Val\right)$.
The semantics is mostly standard: all the non-probabilistic constructs reduce as usual with weight $1$, \eg{}, $\stepdistr(\If \True then \expr_1 \Else \expr_2, \sigma) = \mret(\expr_{1}, \sigma)$ and $\Rand \tapebound$ reduces uniformly at random, \ie{},
\begin{align*}
  & \stepdistr(\Rand \tapebound, \sigma)(n, \sigma) =
    \begin{cases}
      \frac{1}{\tapebound + 1} & \text{for } n \in \{ 0, 1, \ldots, N \} \\
      0 & \text{otherwise}.
    \end{cases}
\end{align*}
We recover the Boolean operation $\Flip$ by defining $\Flip \eqdef{} (\Rand 1 == 1)$.

\subsection{Probabilistic Couplings}\label{sec:couplings}

\emph{Probabilistic coupling} \cite{thorisson/2000, lindvall_lectures_2002, Villani2008OptimalTO} is a mathematical technique for reasoning about pairs of probabilistic processes.
Informally, couplings relate the outputs of two processes by specifying how corresponding sampling statements should be correlated.
This correlation is described by constructing a particular joint distribution over pairs of samples from the two processes.
Traditional definitions of couplings implicitly require that the masses of the two distributions being related are the same.
Instead, we make use of an \emph{asymmetric} notion of couplings, which allows the left-hand side distribution to have less mass than the right-hand side distribution.
\begin{definition}[Left-partial coupling \cite{clutch}]\label{def:coupling}
  Let $\mu_{1} \in \Distr{A}$ and $\mu_{2} \in \Distr{B}$. 
  A sub-distribution $\mu \in \Distr{A \times B}$ is a \emph{left-partial coupling} of $\mu_{1}$ and $\mu_{2}$ if
  \begin{enumerate}
  \item $\All a . \sum_{b \in B} \mu(a, b) = \mu_{1} (a)$
  \item $\All b . \sum_{a \in A} \mu(a, b) \leq \mu_{2} (b)$
  \end{enumerate}
  We write $\refcoupl{\mu_1}{\mu_2}$ if there exists a left-partial coupling of $\mu_{1}$ and $\mu_{2}$.
  Given a relation $R \subseteq A \times B$ we say $\mu$ is a left-partial $R$-coupling if furthermore $\supp(\mu) \subseteq R$.
  We write $\refRcoupl{\mu_1}{\mu_2}{R}$ if there exists a left-partial $R$-coupling of $\mu_{1}$ and $\mu_{2}$.
\end{definition}
Notice how $\refRcoupl{\mathbf{0}}{\mu}{R}$ holds trivially for any $\mu \in \Distr{B}$ and $R \subseteq A \times B$ by picking $\mathbf{0} \in \Distr{A \times B}$ as the witness.
This would not be the case for the traditional symmetric definition of coupling.

Once a coupling has been established, we can often use it to extract a concrete relation between the two probability distributions.
\emph{E.g.}, for $(=)$-couplings, we can conclude point-wise inequality.
\begin{lemma}
  If $\refRcoupl{\mu_{1}}{\mu_{2}}{(=)}$ then $\mu_{1}(a) \leq \mu_{2}(a)$ for all $a$.
\end{lemma}

Moreover---most important for our purposes---from \emph{any} left-partial coupling of $\mu_{1}$ and $\mu_{2}$, we can conclude that the mass of $\mu_{1}$ bounds the mass of $\mu_{2}$ from below.
\begin{lemma}\label{lem:coupling-mass}
  If $\refcoupl{\mu_{1}}{\mu_{2}}$ then $\sum_{a \in A} \mu_{1}(a) \leq \sum_{b \in B} \mu_{2}(b)$.
\end{lemma}

As part of the proof of our soundness theorem (\cref{thm:soundness}), we show that the refinement logic constructs a left-partial coupling of a partial execution of the model and the full execution of the program.
Using \cref{lem:exec-mct} and \cref{lem:coupling-mass} we may then conclude that the termination probability of the program is bounded below by the termination probability of the model.
  

%% file: ref-log.tex
\section{A Probabilistic Termination-Preserving Refinement Logic}\label{sec:rellog}
\newcommand{\walkex}{{\sf walk}}
\thereflog{} is a \emph{probabilistic relational separation logic}.
In this section, we give a high-level overview of \thereflog{}'s rules and walk through simple example uses.

\thereflog{} uses a \emph{refinement} weakest precondition, written $\rwpre{\expr}{\pred}$, where $\expr$ is the program we want to prove a refinement about and $\pred$ is a postcondition.
As in much prior work on refinement reasoning in separation logic, \eg, \citet{caresl,iris-logrel-journal,clutch}, the model that we want to relate to $\expr$ is tracked as a \emph{ghost state} assertion of the form $\spec(\mstate)$, which asserts that the model is currently in state $\mstate$.
We use Markov chains to represent model systems.

To establish a termination-preserving refinement between a Markov chain model starting in a state $\mstate$ and a program $\expr$, we prove an entailment of the form $\spec(\mstate) \vdash \rwpre{\expr}{\pred}$, for an arbitrary postcondition $\pred$.
The following soundness theorem then implies a lower bound between termination of the model and the program:
\begin{theorem}[Soundness]\label{thm:soundness}
  Let $\mstate$ be a state of a Markov chain.
  If \[\spec(\mstate) \proves \rwpre{\expr}{\pred}\] then $\execTerm(\mstate) \leq \execTerm(\expr, \sigma)$ for all program states $\sigma$.
\end{theorem}
For the motivating example discussed in \cref{sec:introduction}, if we label the states of the random walk Markov chain by numbers $n$, then showing $\spec(n) \proves \rwpre{{\sf walk}~n}{\TRUE}$ will establish that ${\sf walk}~n$ is a termination-preserving refinement of the model from starting state $n$. Because the model almost-surely terminates for all $n$, we thus get that ${\sf walk}~n$ almost-surely terminates.

\subsection{Separation Logic Connectives and Basic Unary Rules}
\thereflog{} is developed on top of the Iris framework~\cite{irisjournal} and inherits Iris's basic separation logic connectives. We write $\prop \sep \propB$ for \emph{separating conjunction}, $\prop \wand \propB$ for its adjoint \emph{separating implication} (magic wand), and $\progheap{\loc}{\val}$ for the separation logic resource that denotes ownership of the location $\loc$ containing the value $\val$. Throughout the paper we omit connectives that are used to manipulate Iris-style ghost resources and invariants, \eg{}, the \emph{fancy update modality} \cite{irisjournal} of Iris, since these ideas are orthogonal to the core challenge of refinement reasoning and our use is entirely standard.

For proving refinement weakest preconditions, the logic has typical separation logic rules for reasoning about the basic commands of the language.
A selection of these rules is shown in \cref{fig:program-logic}.
The relation $\expr_{1} \purestep \expr_{2}$ says that $\expr_{1}$ can take a \emph{pure} step, \ie{}, a non-stateful and non-probabilistic reduction step, to the expression $\expr_{2}$.

\begin{figure*}
  \centering
  \begin{align*}
    \expr_{1} \purestep \expr_{2} \ast \rwpre{\expr_{2}}{\pred} &\proves \rwpre{\expr_{1}}{\pred} \ruletag{rwp-pure}\\
    \All \loc . \progheap{\loc}{\val} \wand \pred(\loc) &\proves \rwpre{\Alloc \val}{\pred} \ruletag{rwp-alloc} \\
    (\progheap{\loc}{\val} \wand \pred(\val))\sep {\progheap{\loc}{\val}} &\proves \rwpre{\deref\loc}{\pred} \ruletag{rwp-load} \\
    (\progheap{\loc}{\valB} \wand \pred\TT)\sep {\progheap{\loc}{\val}} &\proves \rwpre{\loc \gets \valB}{\pred} \ruletag{rwp-store} \\
    \All n \leq \tapebound . \pred(n) &\proves \rwpre{\Rand \tapebound}{\pred} \ruletag{rwp-rand} \\
    \pred(\val) &\proves \rwpre{\val}{\pred} \ruletag{rwp-val} \\
    \rwpre{\expr}{\val .\, \rwpre{\fillctx\lctx[\val]}{\pred}} &\proves \rwpre{\fillctx\lctx[\expr]}{\pred} \ruletag{rwp-bind} \\
    (\All \val . \Psi(\val) \wand \pred(\val)) \sep \rwpre{\expr}{\Psi} &\proves \rwpre{\expr}{\pred} \ruletag{rwp-mono} \\
    \prop \sep \rwpre{\expr}{\pred} &\proves \rwpre{\expr}{v . \, \prop \sep \pred(v)} \ruletag{rwp-frame}
  \end{align*}
  \caption{Program logic rules governing the $\rwpre{\expr}{\pred}$ connective.}
  \label{fig:program-logic}
\end{figure*}

\subsection{Guarded Recursion and Relational Rules}
Notably, there is no rule in \cref{fig:program-logic} for reasoning about recursion or loops.
As alluded to in \cref{sec:introduction}, there are multiple ways to encode a form of recursion in a language like \thelang, so rules based on specific syntactic patterns cannot cover the full range of such mechanisms.
Instead, \thereflog{} makes use of \emph{guarded} recursion.

To explain and motivate the need for guarded recursion, let us return to the random walk example from \cref{sec:introduction} and see the issues that arise in trying to prove $\spec(n) \vdash \rwpre{\walkex\, n}{\TRUE}$.

One might first try to prove this by induction on $n$.
However, this attempt would fail in the inductive case, since when the $\Flip$ in $\walkex\ n$ resolves to $\False$, the code effectively makes a recursive call in which the argument is \emph{incremented} to $n+1$.
Thus, in that branch, after stepping through the definition of $\fix$ and $F$, we will eventually find ourselves having to prove a goal of the form $\spec(n+1) \vdash \rwpre{\walkex\, (n+1)}{\TRUE}$, which does not match the induction hypothesis.

Instead, the solution in \thereflog{} is to make use of guarded recursion, in particular the \ruleref{l\"{o}b}
induction rule~\cite{nakano-later}:
\begin{mathpar}
  \inferH{l\"{o}b}
  {\later \prop \proves \prop}
  {\TRUE \proves \prop}
\end{mathpar}
which says that to prove $\prop$, it suffices to prove $\prop$ under the induction hypothesis $\later \prop$, where $\later$ is the so-called ``later'' modality~\cite{nakano-later, appel-later, DBLP:journals/corr/abs-1208-3596}.
Selected other rules for this modality are shown in \cref{fig:later-modality}.
For our example, taking $P$ in the \ruleref{l\"{o}b} rule to be $(\All n. \spec(n) \wand \rwpre{\walkex\, n}{\TRUE})$, will mean that this induction hypothesis allows us to assume the desired refinement holds for \emph{all} $n$, albeit guarded by the later modality.
Because the hypothesis applies for all $n$, we will not run into the obstacle we had with induction on $n$ when the branch in $\Flip$ resolved to $\False$.

\begin{figure}
\begin{mathpar}
  \inferH{later-intro}
  { \prop \proves \propB }
  { \prop \proves \later \propB }
  \and
  \inferH{later-mono}
  { \prop \proves \propB }
  { \later \prop \proves \later \propB }
  \and
  \inferHB{later-and}
  { \prop \proves \later (\propB \land \propC) }
  { \prop \proves \later \propB \land \later \propC }
  \and
  \inferH{later-sep}
  { \prop \proves \later (\propB \ast \propC) }
  { \prop \proves \later \propB \ast \later \propC }
\end{mathpar}
\caption{Selected rules for the $\later$ modality.}
\label{fig:later-modality}
\end{figure}

But how do we eliminate the $\later$ modality guarding the induction hypothesis?
In most guarded program logics, the later modality can be eliminated whenever we take a ``step'' of the program being verified.
However, following \citet{transfinite-iris}, in \thereflog{} the later modality may only be eliminated when the \emph{model} program performs a transition (more precisely, when the model transition system has a transition from the current state of the model to another state).
Since \thereflog{} has no other built-in rule for reasoning about loops or recursion, this ensures that each time the program does some kind of looping or recursion using the \ruleref{l\"{o}b} rule, at least one step will have been performed by the model.
Intuitively, this means that the program can only have a non-terminating execution if the model can take infinitely many transitions.

The simplest later elimination rule applies when the model can make a deterministic transition:
\begin{mathpar}
  \inferH{rwp-spec-det}
  {
    \stepdistr(\mstate_1)(\mstate_2) = 1 \\
    \spec(\mstate_2) \ast \prop \vdash \rwpre{\expr}{\pred}
  }
  {\spec(\mstate_1) \ast \later \prop \vdash \rwpre{\expr}{\pred}}
\end{mathpar}
If the model is currently in a state $\mstate_{1}$, as witnessed by ownership of the resource $\spec(\mstate_{1})$, and the model can deterministically make a step to $\mstate_{2}$, then we may progress the model, stripping a later modality from the assumption $\prop$.

\subsection{Coupling Rules}
Of course, \ruleref{rwp-spec-det} is not sufficient for the example at hand, in which the model only has randomized transitions.
To address this, \thereflog{} also satisfies the following \emph{coupling} rule, similar to that of pRHL~\cite{barthe_formal_2008}, in which the possible transitions of the model and program must be ``matched up'' in a way that preserves probabilities:
\begin{mathpar}
  \inferH{rwp-coupl-rand}
  { \refRcoupl{\stepdistr(\mstate_1)}{\unif(\tapebound)}{R} \\
  \vdash \All (\mstate_2, n) \in R . (\spec(\mstate_2) \ast P) \wand \rwpre{n}{\pred}
  }
  {\spec(\mstate_1) \ast \later P \vdash \rwpre{\Rand \tapebound}{ \pred }}
\end{mathpar}
where $\unif(\tapebound)$ denotes the uniform distribution on the set $\{ 0, 1, \ldots N\}$.
When executing a $\Rand \tapebound$ command, if the model is currently in the state $\mstate_{1}$, the rule says that if we can show a probabilistic coupling $\refRcoupl{\stepdistr(\mstate_1)}{\unif(\tapebound)}{R}$ of the two steps, then we may continue reasoning \emph{as if} the program and the model progressed to states in the support of the coupling.
Furthermore, the $\later$ guarding the assumption $P$ is removed, reflecting that the model has made a transition.

As an example, a special case of this rule for $\Flip$ is the following:
\begin{mathpar}
  \inferrule{
    m_f \neq m_t \\\\
    \stepdistr(m)(m_f) = 1/2 \\
    P \ast \spec(m_f) \vdash \rwpre{K[\False]}{\pred} \\\\
    \stepdistr(m)(m_t) = 1/2 \\
    P \ast \spec(m_t) \vdash \rwpre{K[\True]}{\pred}
  }
  {\later P \ast \spec(m) \vdash \rwpre{K[\Flip]}{\pred}}
\end{mathpar}
With this rule, we have to reason about two cases for how the $\Flip$ command resolves, and for each case, we pick a state to transition the model to, subject to the requirement that the model transitions happen with the same probability 1/2 as the transitions that $\Flip$ makes.

In the $\walkex$ example, when the model is in state $m = n$, we take $m_f = n+1$ and $m_t = n-1$.
Thus, for the case where the $\Flip$ resolves to $\False$ and then makes the recursive call of $\walkex\ (n+1)$, we will have $\spec(n+1)$.
Therefore, we may make use of the induction hypothesis from L\"{o}b induction, which will have had the $\later$ modality removed through the application of the rule for $\Flip$.
The case for $\True$ is similar.

\subsection{Example: Repeated Coin Flips}

We now consider another example of reasoning about loops in \thereflog{}. One of the simplest almost-surely (but not always) terminating stochastic processes is the one that repeatedly tosses a fair coin until it gets tails.
There exists a non-terminating run, \ie{}, the run where one always gets heads, but it happens only with probability zero.
We can describe the process using the following diagram where $\True$ is used for heads and $\False$ for tails.
\begin{center}
  \tikzset{every loop/.style={min distance=5mm,in=180,out=140,looseness=2}}
  \begin{tikzpicture}[node distance=1.5cm,on grid, auto]
    \node[state] (true) {$\True$};
    \node[state, accepting] [right=of true] (false) {$\False$};
    \path[->] (true) edge [loop above] node {$\frac{1}{2}$} (true);
    \path[->] (true) edge [bend right] node[below] {$\frac{1}{2}$} (false);
  \end{tikzpicture}
\end{center}
It is straightforward to show that $\exec_{n}(\True) = 1 - \frac{1}{2}^{n}$ and thus $\execTerm(\True) = 1$.

Using the notation $\WhileS \expr_{1} do \expr_{2} end \eqdef{} (\Rec f () = \If \expr_{1} then \expr_{2}; f~() \Else ())~()$ we implement this process as a loop that repeatedly flips a coin until it gets $\False$.
\newcommand{\flips}{\mathsf{flips}}
\begin{align*}
  \flips \eqdef{} &\WhileS \Flip do () end
\end{align*}
Using \thereflog{}, we show that $\flips$ formally refines the model by showing the specification
\[
  \vdash \spec(\True) \wand \rwpre{\flips}{\_ .~\spec(\False)}
\]
and thus $\flips$ almost-surely terminates.

To prove the entailment we apply \ruleref{l\"{o}b} and are left with the proof obligation
\[
  \later\left(\spec(\True) \wand \rwpre{\flips}{\_ .~\spec(\False)} \right) \proves \spec(\True) \wand \rwpre{\flips}{\_ .~\spec(\False)}.
\]
That is, we have assumed our initial goal but under a later modality.
After introducing the specification resource, we symbolically step the program forward using \ruleref{rwp-pure}.
We now apply \ruleref{rwp-coupl-rand} using an equality coupling of $\unif(1)$ which allows us to strip the later modality from the induction hypothesis and continue reasoning \emph{as if} the coin in the program and the model have the same outcome $b$.
\begin{align*}
  &\spec(b) \sep \left(\spec(\True) \wand \rwpre{\flips}{\_ .~\spec(\False)} \right) \proves \\
  &\rwpre{\If b then ();~\flips \Else ()}{ \_ .~\spec(\False)}.
\end{align*}
We continue by case distinction on $b$: if $b$ is $\False$ we are immediately done and if $b$ is $\True$ we apply \ruleref{rwp-pure} followed by the induction hypothesis which finishes the proof.

While the repeated coin flip is a simple example that mainly serves to illustrate a minimal proof in \thereflog{}, the same pattern and recipe applies to many other first-order probabilistic looping constructs as well.
For example, the program
\begin{align*}
  &\Let r = \Alloc n in \\[-3pt]
  & \While \deref r \neq 0 do
    {
    \If \Flip then r \gets \deref r - 1 \Else r \gets \deref r + 1
    }
    end {}
\end{align*}
can be shown to refine the symmetric random walk from \cref{sec:introduction} using exactly this pattern and the refinement proof looks much like for the repeated coin flip.

\subsection{Summary}
As we have seen, at a high level, the rules of \thereflog{} combine three key ingredients:
\begin{enumerate}
\item \emph{Higher-order separation logic}, as embodied in the Iris framework~\cite{irisjournal}, which provides powerful tools for reasoning about modular use of state in higher-order programs.

\item \emph{Guarded recursion}, formulated as in recent works on termination-preserving refinement~\cite{transfinite-iris} to ensure that looping in the program is matched with transitions in the model.

\item \emph{Couplings}, as in pRHL~\cite{barthe_formal_2008}, which allow for ``aligning'' the probabilistic transitions of the program and model.
\end{enumerate}
In the end, the logic may seem surprisingly---and perhaps suspiciously---simple, but this simplicity stems from the use of these powerful and well-tested abstractions.
The main challenge and novelty of \thereflog{}, then, lies in showing that this combination of rules is sound for proving probabilistic termination-preserving refinements, the question that we will turn to in \cref{sec:soundness}.
Before doing so, we first examine a way to make the coupling rule more flexible (\cref{sec:async}) and then explore a number of case studies using \thereflog{} (\cref{sec:examples}).


%% file: soundness.tex
\section{Semantic Model and Soundness}\label{sec:soundness}

The soundness of \thereflog{} is justified by constructing a semantic model of $\rwpre{\expr}{\pred}$ in terms of the Iris ``base logic''~\cite{irisjournal}. 
This underlying base logic is a separation logic with the various connectives we saw in earlier sections, including the $\later$ modality and the L\"{o}b induction principle.
This section defines the semantic model and then shows how the model implies \cref{thm:soundness}.
For notational convenience, we write an inference rule with premises $\prop_{1}, \ldots, \prop_{n}$ and conclusion $\propB$ as notation for $(\prop_{1} \sep \ldots \sep \prop_{n}) \proves Q$ in the Iris base logic.

\subsection{Model}

The semantic model of the refinement weakest precondition $\rwpre{\expr}{\pred}$ constructs a coupling of the execution of $\mstate$ as tracked by the $\spec(\mstate)$ resource and the execution of the program $\expr$.
Intuitively, it does so by constructing individual stepwise couplings as the proof symbolically executes the program and the model.
In the end, these stepwise couplings will all be combined to construct a coupling of the full executions.

In contrast to many logics making use of guarded recursion, $\rwpre{\expr}{\pred}$ is defined as a \emph{least} fixed point.
This is reminiscent of models for weakest preconditions that ensure total program correctness but our simultaneous use of guarded recursion will permit non-termination in a controlled way. 
Formally, the least fixed point $\lfp x . t$ exists if $t$ is monotone, \ie{}, all recursive occurrences of $x$ appear in a positive position, as follows from Tarski's fixed-point theorem \cite{Tarski_1955}.
The definition looks as follows.
\begin{align*}
  \rwpre{\expr_{1}}{\pred} \eqdef{} \lfp W .
  & (\expr_{1} \in \Val \sep \pred(\expr_{1})) \lor{} \\
  & (\expr_{1} \not\in \Val \sep \All \mstate_{1}, \state_{1} . \specinterp(\mstate_{1}) \sep \stateinterp(\state_{1}) \wand{} \\
  & \quad \execCoupl{\mstate_{1}}{(\expr_{1}, \state_{1})}{\mstate_{2}, (\expr_{2}, \state_{2}) .\, \specinterp(\mstate_{2}) \sep \stateinterp(\state_{2}) \sep W(\expr_{2}, \pred)})
\end{align*}
The left disjunct of the definition says that if $\expr_{1}$ is a value, then the postcondition $\pred(\expr_{1})$ must hold.
Meanwhile, the right side says if $\expr_{1}$ is \emph{not} a value, then we get to assume ownership of two resources $\specinterp(\mstate_{1})$ and $\stateinterp(\state_{1})$ and have to prove a \emph{coupling precondition} $\execCoupl{\mstate_{1}}{(\expr_{1}, \state_{1})}{\ldots}$ as defined below.
The postcondition of the coupling precondition requires the prover to give back the two updated resources and show that $\rwpre{\expr_{2}}{\pred}$ holds recursively.

The two resources $\specinterp(\mstate_{1})$ and $\stateinterp(\state_{1})$ are, respectively, a model and a state \emph{interpretation}.
Formally, they track \emph{authoritative} views of the model and the state~\cite{iris}.
The model interpretation always agrees with the model state tracked with the $\spec(\mstate)$ resource, \ie{}, $\specinterp(\mstate) \sep \spec(\mstate') \proves \mstate = \mstate'$, and the state interpretation always agrees with the points-to connective $\progheap{\loc}{\val}$ for the heap, \ie{}, $\stateinterp(\state) \sep \progheap{\loc}{\val} \proves \state(\loc) = \val$.

The coupling precondition is the heart of the probabilistic program logic and ensures (1) the \emph{safety} of $(\expr_{1}, \state_{1})$, meaning that it does not get stuck, and (2) the existence of a relational coupling with the model.
The connective $\execCoupl{\mstate}{\cfg}{\predB}$ is a ternary relation on a model state $\mstate \in \Markov$, a program configuration $\cfg \in \Conf$, and a relational post condition $\predB : \Markov \to \Conf \to \iProp$ where $\iProp$ is the type of propositions in the logic.
Intuitively, it forms a relational coupling logic that establishes the existence of a probabilistic coupling of \emph{one} step of the program configuration $\cfg$ with a number of steps of the model $\mstate$ such that the postcondition $\predB$ holds for the support.
Formally, it is defined inductively by the inference rules shown in \cref{fig:cpre} (\ie{} as a least fixed point).

\begin{figure*}
  \centering
  \begin{mathpar}
    \inferH{cpl-prog}
    { \red(\cfg_1) \\
      \refRcoupl{\mret(\mstate)}{\stepdistr(\cfg_{1})}{R} \\
      \All s \in R . \predB(s)
    }
    {\execCoupl{\mstate}{\cfg_{1}}{\predB}}
    \and
    \inferH{cpl-model-prog}
    { \red(\cfg_1) \\
      \red(\mstate_1) \\
      \refRcoupl{\stepdistr(\mstate_1)}{\stepdistr(\cfg_1)}{R} \\
      \All s \in R . \later \predB(s)
    }
    {\execCoupl{\mstate_1}{\cfg_1}{\predB}}
    \and
    \inferH{cpl-model}
    { \red(\mstate_1) \\
      \refRcoupl{\stepdistr(\mstate_{1})}{\mret(\cfg)}{R} \\
      \All (\mstate_{2}, \cfg) \in R . \later \execCoupl{\mstate_2}{\cfg}{\predB}
    }
    {\execCoupl{\mstate_1}{\cfg}{\predB}}
  \end{mathpar}
  \caption{Inductive definition of the coupling precondition $\execCoupl{\mstate}{\cfg}{\predB}$.}
  \label{fig:cpre}
\end{figure*}

The first constructor \ruleref{cpl-prog} applies to symbolic steps that only progress the program.
It requires the configuration to be reducible and that there is a trivial coupling between the Dirac distribution of the model state $\mret(\mstate)$ and the program step, which just means that the program configuration can take a step.
Moreover, everything in the support of the coupling must satisfy the postcondition.
The constructor is essential to validating all the unary program logic rules shown in \cref{fig:program-logic}.

The second constructor \ruleref{cpl-model-prog} is used to validate \ruleref{rwp-coupl-rand}, the coupling rule in the logic.
It requires that the model and the configuration are reducible (to guarantee safety) and that a coupling can be exhibited between a step of the model and a step of the program.
Finally, everything in the support of the coupling must satisfy $\predB$ but \emph{under a later modality}.
This occurrence is one of the two places that formally connects the later modality to steps of the model.

The third and last constructor \ruleref{cpl-model} is used to validate \ruleref{rwp-spec-det}, and thus the symbolic steps which only progress the model.
It requires a trivial coupling of the Dirac distribution of the program configuration with the model step, which intuitively just means that the model can take a step.
For everything in the support of the coupling, the coupling precondition must hold recursively \emph{under a later modality}, again connecting the later modality to steps of the model.

The definition of the refinement weakest precondition consists of multiple interacting components: probabilistic couplings, later modalities, resources, and two fixed points.
It is this fine (and subtle!) balance of the components that allows us in the following section to prove that termination is indeed preserved across the program and the model, but it is also what allows us to enable the reasoning principles that we want.
For example, the fact that \emph{one} unfolding of the $\rwpre{\expr}{\pred}$ fixed point always corresponds to \emph{one} step of the program (but possibly multiple steps of the model) is crucial to stating and proving soundness of the rules in \cref{fig:program-logic}.

\subsection{Soundness}
We show soundness of \thereflog{} in two stages:

\begin{enumerate}
\item we show that the relational logic establishes a so-called ``plain'' guarded refinement $\refines{\mstate}{\cfg}$, \ie{} a guarded relation that does not depend on separation logic resources, and
\item we show that plain guarded refinement implies preservation of termination.
\end{enumerate}
The plain guarded refinement is defined inductively by the rules in \cref{fig:plain-ref}.
If the program has terminated the refinement trivially holds (\ruleref{ref-val}), we can step the program and the model independently (\ruleref{ref-prog} and \ruleref{ref-model}, respectively), and we can incorporate non-trivial couplings of model and program steps (\ruleref{ref-model-prog}).
In the two cases where we progress the model (\ruleref{ref-model} and \ruleref{ref-model-prog}), the recursive occurrence of the refinement is under a later modality, thus connecting the later modality to steps of the model as in the definition of the relational logic.

\begin{figure*}
  \centering
  \begin{mathpar}
    \inferH{ref-val}
    {\val \in \Val}
    {\refines{\mstate}{(\val, \state)}}
    \and
    \inferH{ref-prog}
    { \red(\cfg_1) \\
      \refRcoupl{\mret(\mstate)}{\stepdistr(\cfg_1)}{R} \\
      \All (\mstate, \cfg_2) \in R . \refines{\mstate}{\cfg_2}
    }
    {\refines{\mstate}{\cfg_1}}
    \and
    \inferH{ref-model-prog}
    { \red(\cfg_1) \\
      \red(\mstate_1) \\
      \refRcoupl{\stepdistr(\mstate_1)}{\stepdistr(\cfg_2)}{R} \\
      \All (\mstate_2, \cfg_2) \in R . \later \refines{\mstate_2}{\cfg_2}
    }
    {\refines{\mstate_1}{\cfg_1}}
    \and
    \inferH{ref-model}
    { \red(\mstate_1) \\
      \refRcoupl{\stepdistr(\mstate_1)}{\mret(\cfg)}{R} \\
      \All (\mstate_2, \cfg) \in R . \later \refines{\mstate_2}{\cfg}
    }
    {\refines{\mstate_1}{\cfg}}
  \end{mathpar}
  \caption{Inductive definition of the (plain) guarded refinement relation $\refines{\mstate}{\cfg}$.}
  \label{fig:plain-ref}
\end{figure*}

The first stage of the soundness proof is the following lemma.
\begin{lemma}\label{lem:rwp-refines}
  If $\spec(\mstate) \proves \rwpre{\expr}{\pred}$ then $\proves \refines{\mstate}{(\expr, \state)}$ for all $\state$.
\end{lemma}
The proof goes by structural induction in both the weakest precondition $\rwpre{\expr}{\pred}$ and the coupling precondition $\execCoupl{\mstate}{\cfg}{\predB}$ fixed points.
While the details of how resources are erased depends on how they are managed in Iris (\eg{}, through the fancy update modality, which we have omitted), for intuition about why this should hold, observe that \emph{if} we ignore the model and state interpretation resources, each case of $\rwpre{\expr}{\pred}$ and the constructors of $\execCoupl{\mstate}{\cfg}{\predB}$ correspond exactly to one constructor of the $\refines{\mstate}{\cfg}$ refinement relation.

The core of the soundness proof is the second stage and the fact that the $\refines{\mstate}{\cfg}$ relation preserves termination.
The key enabler is the monotone convergence of termination probability (\cref{lem:exec-mct}), that is, to show that the termination probability of the program is bounded below by the termination probability of the model, it suffices to consider all finite prefixes of the model execution---exactly what our guarded refinement relation is concerned with.
This in turn means that to show that termination is preserved, we ``just'' have to combine the stepwise left-partial couplings constructed in the refinement relation into a single coupling of executions.
The ability to combine couplings in this way follows from the following lemma showing that left-partial couplings can be composed along the monadic structure of sub-distributions:
 \begin{lemma}[Composition of couplings]\label{lem:coupling-comp}
   Let $R \subseteq A \times B$, $S \subseteq A' \times B'$, $\mu_{1} \in \Distr{A}$, $\mu_2 \in \Distr{B}$, $f_{1} : A \to \Distr{A'}$, and $f_{2} : B \to \Distr{B'}$.
   \begin{enumerate}
   \item If $(a, b) \in R$ then $\refRcoupl{\mret(a)}{\mret(b)}{R} $.
   \item If $\refRcoupl{\mu_{1}}{\mu_{2}}{R}$ and for all $(a, b) \in R$ it is the case that $\refRcoupl{f_{1}(a)}{f_{2}(b)}{S}$ then $\refRcoupl{\mu_{1} \mbindi f_{1}}{\mu_{2} \mbindi f_{2}}{S}$
   \end{enumerate}
 \end{lemma}

Using this lemma, we have the following:
\begin{lemma}\label{lem:refines-soundness-laterN}
  $\refines{\mstate}{\cfg} \proves \later^{n} \refcoupl{\exec_{n}(\mstate)}{\exec(\cfg)}$ for all $n$.
\end{lemma}
\begin{proof}
  The proof proceeds by induction on the $\refines{\mstate}{\cfg}$ fixed point.

  \begin{description}
  \item[Case \ruleref{ref-val}.]
    Since $\rho = (\val, \state)$ and $\val \in \Val$ we get that $\exec(\val, \state) = \mret(\val)$.
    As $\refcoupl{\mu}{\mret(\val)}$ trivially holds for any $\mu$ (pick a coupling that relates all the mass of $\mu$ to $\val$), we conclude.
  \item[Case \ruleref{ref-prog}.]
    Since $\cfg$ is reducible, we get that $\exec(\cfg) = \stepdistr(\cfg) \mbindi \exec$.
    By the left identity law of the distribution monad, $\exec_{n}(\mstate) = \mret(\mstate) \mbindi \exec_{n}$.
    We are left with the goal
    \[\textstyle\later^{n} (\refcoupl{\mret(\mstate) \mbindi \exec_{n}}{\stepdistr(\cfg) \mbindi \exec}). \]
    Using \ruleref{later-mono} we apply \cref{lem:coupling-comp} under the later modalities and exploit the coupling $\refRcoupl{\mret(\mstate)}{\stepdistr(\cfg_1)}{R}$ which leaves us with the goal
    \[\textstyle\later^{n} \All (\mstate, \cfg_{2}) \in R . \refcoupl{\exec_{n}(\mstate)}{\exec(\cfg_{2})} \]
    which follows by the induction hypothesis.
  \item[Case \ruleref{ref-model-prog}.]
    We do a case distinction on $n$.
    If $n = 0$ then $\exec_{0}(\mstate) = \mathbf{0}$ and thus the left-partial coupling exists trivially.
    If $n \neq 0$ then $\exec_{n}(\mstate) = \stepdistr(\mstate) \mbindi \exec_{n - 1}$ and since $\cfg$ is reducible, we get that $\exec(\cfg) = \stepdistr(\cfg) \mbindi \exec$.
    This leaves us with the goal
    \[\textstyle\later^{n} (\refcoupl{\stepdistr(\mstate) \mbindi \exec_{n-1}}{\stepdistr(\cfg) \mbindi \exec}) \]
    which follows as above by \ruleref{later-mono}, \cref{lem:coupling-comp}, and the induction hypothesis.
  \item[Case \ruleref{ref-model}.]
    We do a case distinction on $n$.
    If $n = 0$ then $\exec_{0}(\mstate) = \mathbf{0}$ and thus the left-partial coupling exists trivially.
    If $n \neq 0$ then $\exec_{n}(\mstate) = \stepdistr(\mstate) \mbindi \exec_{n - 1}$ and by the left identity law $\exec(\cfg) = \mret(\cfg) \mbindi \exec$.
    This leaves us with the goal
    \[\textstyle\later^{n} (\refcoupl{\stepdistr(\mstate) \mbindi \exec_{n-1}}{\mret(\cfg) \mbindi \exec}) \]
    which follows as above by \ruleref{later-mono}, \cref{lem:coupling-comp}, and the induction hypothesis.
  \end{description}
  \vspace{-1.2em}
\end{proof}

The result shows that the desired coupling exists, but this existence is \emph{internal} to the Iris base logic, and under $n$ iterations of the $\later$ modality.
At this point we rely on the following soundness theorem for the Iris base logic to know that the coupling exists externally in the meta-logic:
\begin{theorem}\label{thm:si-soundness}
  Let $\varphi$ be a meta-logic proposition.
  If $\proves \later^n \varphi$ then $\varphi$ holds in the meta-logic.
\end{theorem}
\begin{corollary}\label{cor:refines-soundness}
  If $\proves \refines{\mstate}{\cfg}$ then $\refcoupl{\exec_{n}(\mstate)}{\exec(\cfg)}$ for all $n$.
\end{corollary}
\begin{proof}
  Immediate by applying \cref{thm:si-soundness} and \cref{lem:refines-soundness-laterN}.
\end{proof}
\begin{corollary}\label{cor:rwp-refcoupl}
  If $\spec(\mstate) \proves \rwpre{\expr}{\pred}$ then $\refcoupl{\exec_{n}(\mstate)}{\exec(\expr, \state)}$ for all $n$ and $\state$.
\end{corollary}
\begin{proof}
  Immediate by applying \cref{cor:refines-soundness} and \cref{lem:rwp-refines}.
\end{proof}

The soundness theorem of \thereflog{} (\cref{thm:soundness}) then follows directly by applying \cref{lem:exec-mct}, \cref{{lem:coupling-mass}}, and \cref{cor:rwp-refcoupl}.

\paragraph{Asynchronous couplings}
To incorporate asynchronous couplings into \thereflog{}, we add a fourth constructor to the coupling precondition that adds the possibility of coupling a model step with any number of presampling steps.
The distribution $\foldM(\statestepdistr, \state_{1}, l)$ denotes a monadic fold of $\statestepdistr$ over the list $l$ of tape labels using $\state_{1}$ as the initial value.
\begin{mathpar}
  \infer
  { \red(\mstate_1) \\
    l \subseteq \dom(\state_1) \\
    \refRcoupl{\stepdistr(\mstate_{1})}{\foldM(\statestepdistr, \state_{1}, l)}{R} \\
    \All (\mstate_{2}, \state_{2}) \in R . \later \execCoupl{\mstate_2}{(\expr_1, \state_2)}{\predB}
  }
  {\execCoupl{(\expr_1, \state_1)}{\mstate_1}{\predB}}
\end{mathpar}
The state interpretation is extended accordingly to give meaning to the $\progtape{\lbl}{\tapebound}{\tape}$ resource.

As for the coupling precondition we also extend the plain guarded refinement relation in \cref{fig:plain-ref}.
The soundness theorem can then adapted by making use of \cref{lem:erasure} to erase presampling steps.

\subsection{Comparison to Guarded Recursion for Non-Probabilistic Termination Preservation} \label{sec:sound-compare}

As mentioned earlier, prior works have also built program logics for termination-preserving refinements using
guarded recursion~\citep{DBLP:conf/esop/TassarottiJ017, trillium, transfinite-iris}.
These works target non-probabilistic languages that instead have (adversarial) non-determinism.

In terms of logical rules, these works use similar core mechanisms of (1) representing a specification program or model as ghost state, (2) reasoning about loops using L\"{o}b induction, and (3) only allowing the later modality to be eliminated when the specification program takes a step.
\thereflog{} differs primarily in that its coupling rule requires the resolution of probabilistic non-determinism between the program and the model to have corresponding probabilities, whereas in these prior works, non-determinism at the model level is resolved angelically.
Of course, the resulting soundness theorems for these prior logics also differ.
They say that if the program has a non-terminating execution, then the model must also have a non-terminating execution.\footnote{\citet{DBLP:conf/esop/TassarottiJ017} and \citet{trillium} consider a concurrent language and allow for a stronger property, requiring that if the non-terminating execution in the program was under a \emph{fair} scheduler, then the execution in the model must also be fair. This imposes some restrictions on how non-determinism at the model is resolved, so that it is not entirely angelic.}

The other key difference is that to prove their soundness theorems, these prior works have found it necessary to either move to transfinite step indexing~\citep{DBLP:conf/esop/SvendsenSB16}, or to require some kind of \emph{finiteness} condition, \eg require that models be \emph{finitely branching} (meaning that each state can only move to finitely many states in a single transition) or relative image-finiteness of the refinement relation.
\thereflog{}'s soundness proof requires no such restriction.
The model uses ``standard'' natural number step indexing, and a Markov chain model is allowed to have countable branching, since \cref{def:markov-chain} permits the support of the chain's $\stepdistr$ function to be countable.

One might wonder why these technical workarounds were not needed in the probabilistic case, and whether \thereflog{}'s soundness proof implies that something similar could be done in the non-deterministic case as well.
The answer lies in \cref{lem:exec-mct}, which shows that the termination probability of a program can be lower-bounded by considering the termination probability of its $n$-step finite approximations.
This theorem was used in the soundness proof of \thereflog{}, allowing us in \cref{cor:rwp-refcoupl} to consider executions of up to $n$ steps of the model for each $n$.
Only considering $n$-step executions was important because these corresponded to the up to $n$ iterations of the later modality incurred when unfolding the guarded recursion defining the coupling precondition.

A corresponding proof approach does not work in the context of non-probabilistic termination, because there is no useful analogue of \cref{lem:exec-mct}.
In general, knowing that for all $n$, a non-deterministic program has an execution that has not terminated after $n$ steps does not imply that it necessarily has a diverging execution.
Thus, the soundness proofs in the aforementioned logics cannot use the approach of considering $n$ step unfoldings of guarded recursion that we used above.


%% file: asynchronous.tex
\section{Asynchronous Couplings}\label{sec:async}

Using the probabilistic coupling rules we have seen so far requires aligning or ``synchronizing'' the sampling statements of the two probabilistic processes being related.
For example, both the program and its model have to be executing the sample statements we want to couple for their next step when applying rules like \ruleref{rwp-coupl-rand}.
However, it is not always possible to synchronize sampling statements in this way, especially when considering higher-order programs.
To address this issue, \citet{clutch} introduce \emph{asynchronous coupling} for proving contextual refinement of (higher-order) probabilistic programs.
We identify two new and orthogonal use cases for asynchronous coupling in \thereflog{} which address the following two issues:
\begin{enumerate}
\item When relating a complex program to a simpler model, it is sometimes necessary to couple \emph{one} model step to \emph{multiple} non-adjacent program samplings (as illustrated in \cref{sec:lazy-real}).
\item Sometimes a later modality needs to be eliminated \emph{now}, but the coupling step---which would introduce the later modality---only happens \emph{in the future} (as illustrated in \cref{sec:gwt}).
\end{enumerate}
In the remainder of this section, we recall the concept of asynchronous coupling and describe how it is incorporated into \thereflog{}.
The reader may want to initially skip this section but return before reading \cref{sec:lazy-real} and \cref{sec:gwt}.

\paragraph{Presampling tapes.}
Asynchronous couplings are introduced through dynamically-allocated \emph{presampling tapes}.
Intuitively, presampling tapes will allow us \emph{in the logic} to presample (and in turn couple) the outcome of future sampling statements.

Formally, presampling tapes appear as two new constructs added to the programming language.
\begin{align*}
  \expr \in \Expr \bnfdef{}& \ldots \ALT
                             \AllocTape~\expr \ALT
                             \Rand \expr_{1}~\expr_{2}
\end{align*}
The $\AllocTape~\tapebound$ operation allocates a new fresh tape with the upper bound $\tapebound$, representing future outcomes of $\Rand \tapebound$ operations.
The $\Rand$ primitive can now (optionally) be annotated with a tape label $\lbl$.
If the corresponding tape is empty, $\Rand~\tapebound~\lbl$ reduces to any $n \leq \tapebound$ with equal probability, just as if it had not been labeled.
But if the tape is \emph{not} empty, then $\Rand~\tapebound~\lbl$ reduces \emph{deterministically} by taking off the first element of the tape and returning it.
However, \emph{no} primitives in the language will add values to the tapes.
Instead, values are added to tapes as part of presampling steps that will be \emph{ghost operations} appearing only in the logic.
In fact, labeled and unlabeled samplings operations are contextually equivalent \cite{clutch}.

At the logical level, presampling tapes come with a $\progtape{\lbl}{\tapebound}{\tape}$ assertion that denotes \emph{ownership} of the label $\lbl$ and its contents $(\tapebound, \tape)$, analogously to how the traditional points-to-connective $\progheap{\loc}{\val}$ of separation logic denotes ownership of the location $\loc$ and its contents on the heap.
When a tape is allocated, ownership of a fresh empty tape is acquired, \ie{}
\begin{align*}
  \All \lbl . \progtape{\lbl}{\tapebound}{\nil} \wand \pred(\lbl) & \proves \rwpre{\AllocTape \tapebound}{\pred} \ruletag{rwp-tape-alloc}
\end{align*}
If one owns $\progtape{\lbl}{\tapebound}{\nil}$, \ie{}, when the corresponding tape is empty, then $\Rand~\tapebound~\lbl$ reduces symbolically to any $n \leq \tapebound$, reflecting the operational behavior described above:
\begin{align*}
  (\All n \leq \tapebound . \progtape{\lbl}{\tapebound}{\nil} \wand \pred(n)) \sep \progtape{\lbl}{\tapebound}{\nil} &\proves \rwpre{\Rand \tapebound~\lbl}{\pred} \ruletag{rwp-tape-empty}
\end{align*}
When the tape is \emph{not} empty, then $\Rand~\tapebound~\lbl$ reduces symbolically by taking off the first element of the tape and returning it.
\begin{align*}
  (\progtape{\lbl}{\tapebound}{\tape} \wand \pred(n)) \sep \progtape{\lbl}{\tapebound}{n \cons \tape} &\proves \rwpre{\Rand \tapebound~\lbl}{\pred} \ruletag{rwp-tape}
\end{align*}
Asynchronous couplings can now be introduced in the logic by coupling rules that couple \emph{any} finite number of presampling steps onto tapes with a model step.
When we---at some point in the future---reach a presampled $\Rand~\tapebound~\lbl$ operation, we simply read off the presampled values from the $\lbl$ tape deterministically in a first-in-first-out order.
For example, an asynchronous variant of the coupling rule for $\Flip$ originally shown in \cref{sec:rellog} is the following:
\begin{align*}
  \inferrule{
    m_f \neq m_t \\\\
    \stepdistr(m)(m_f) = 1/2 \\
    P \sep \spec(m_f) \sep \progtape{\lbl}{1}{\vec{b} \cdot \False} \vdash \rwpre{\expr}{\pred} \\\\
    \stepdistr(m)(m_t) = 1/2 \\
    P \ast \spec(m_t) \sep \progtape{\lbl}{1}{\vec{b} \cdot \True} \vdash \rwpre{\expr}{\pred}
  }
  {\later P \sep \progtape{\lbl}{1}{\vec{b}} \sep \spec(m) \vdash \rwpre{\expr}{\pred}}
\end{align*}
Instead of reasoning about two cases for how a $\Flip$ operation resolves, we reason about two cases for how a Boolean is sampled onto the tape $\lbl$.
This pattern can be generalized to allow one model step to be coupled with multiple presampling steps, which we exploit in \cref{sec:lazy-real}.
Notice moreover that since a model step is taken, we are also allowed to eliminate a later modality from the assumption $\prop$ \emph{before} reaching the $\Flip$ operation.
This fact will be crucial for the example considered in \cref{sec:gwt}.

Soundness of asynchronous couplings hinges on the fact that presampling operationally \emph{does not matter} as seen by  the erasure theorem below.
The distribution $\statestepdistr(\state_{1}, \lbl)$ appends a uniformly sampled value to the end of the $\lbl$ tape in $\state_{1}$.
\begin{lemma}[Erasure]\label{lem:erasure}
  If $\lbl \in \dom(\state_{1})$ then
  $ \exec(\expr, \state_{1}) =  \statestepdistr(\state_{1}, \lbl) \mbindi (\Lam \state_{2} . \exec(\expr, \state_{2}))$.
\end{lemma}


%% file: examples.tex
\section{Case Studies}\label{sec:examples}
\newcommand{\unittyalt}{\textdom{unit}}
In this section, we develop a series of case studies of increasing complexity both in terms of program size and also in terms of proof complexity.
The examples we present are chosen not only to illustrate how \thereflog{} is applied, but also to demonstrate how working in higher-order guarded separation logic allows for concise and composable specifications.
For the sake of presentation, we make use of standard~\cite[\S6]{irisjournal} Hoare-triple notation
\[
  \hoare{\prop}{\expr}{\val . \propB} \eqdef{} \always (\prop \wand \rwpre{\expr}{\val . \propB})
\]
throughout this section.
The use of the \emph{persistence} modality $\always \prop$ says that $\prop$ holds without asserting any exclusive ownership of resources and thus that it can be arbitrarily duplicated.

\subsection{Recursion Through the Store}
Recall the motivating example from the introduction, which uses Landin's knot to define a fixed-point combinator {\sf fix} and then applies {\sf fix} to define a recursive randomized program.
\begin{align*}
  {\sf fix} &\triangleq \Lam f . \Let r = \Alloc~(\Lam x. x) in r \gets (\Lam x. f~(\deref{r})~x) ;~\deref{r}
\end{align*}
In our walk through of this example earlier in \cref{sec:rellog}, we elided a discussion of how verification of the code making up {\sf fix} itself works.
In fact, as a first step, we may show a general higher-order specification for the fixed-point combinator ${\sf fix}$.
For all abstract predicates $\pred, \predB : \Val \to \iProp$, where $\iProp$ is the type of propositions in the logic, we show the specification
\begin{align*}
  \left(\All f, \val' . \hoare{\All \val'' . \later \left(\hoare{\pred(\val'')}{f~\val''}{\predB} \right)}{F~f~\val'}{\predB} \ast \pred(\val')\right) \proves \hoare{\pred(\val)}{{\sf fix}~F~\val}{\predB}.
\end{align*}
The specification says that to prove postcondition $\predB$ of ${\sf fix}~F~\val$ given precondition $\pred(\val)$, it suffices to show a specification for $F$ with postcondition $\predB$.
In proving the specification of $F$, however, one may assume that the first argument $f$ (used for recursive calls) satisfies the specification as well, but under a later modality.
Notice that the specification does not say anything explicitly about refinement---in fact, the same specification is given to ${\sf fix}$ in logics for partial correctness (see, \eg{}, \citet{iris-tutorial}) but \emph{without} the later modality.
In our specification, the later modality signifies an obligation to take a model step: intuitively, to recurse (and hence potentially not terminate), at least one step of the model must be exhibited in order for termination to be preserved.
The proof of the specification is essentially identical to a proof carried out in standard Iris: after symbolically evaluating the store operation $\deref r$ one applies \ruleref{l\"{o}b} and the specification then follows by the assumed specification of $F$.

The fact that ${\sf walk}~n$ refines the symmetric random walk model follows by showing the specification $\hoare{\spec(n)}{{\sf walk}~n}{ \_\, . \Exists m. \spec(m)}$ which follows by applying the higher-order specification of ${\sf fix}$, picking $\pred(n) \eqdef \spec(n)$ and $\predB(\val) \eqdef{} \Exists m . \spec(m)$.
Instead of applying \ruleref{l\"{o}b} induction directly, as we did when initially describing the example, we apply the specification of $f$ as provided by ${\sf fix}$ to recurse.
A key benefit of this approach is that the use of a higher-order specification allows for a refinement proof that is agnostic to the intricacies of how iteration is realized.
As a result, the proof can be re-used for different implementations of the recursor.

\subsection{List Generators}
\citet{DBLP:journals/lmcs/KobayashiLG20} study probabilistic programs with higher-order functions (but without state) using probabilistic higher-order recursion schemes.
As a motivating example, they present the following program that combines probabilistic choice and higher-order functions.
\newcommand{\listgen}{{\sf listgen}}
\begin{align*}
  \Rec \listgen f = \spac
  &\If \Flip then \None \\
  &\Else \Let h = f~() in \\
  &\phantom{\Else} \Let t = \listgen~f in \\
  &\phantom{\Else} \Some~(h, t)
\end{align*}
The function $\listgen$ takes a generator $f$ of elements as an argument and creates a list of elements, each of them obtained by calling $f$.
The length of the list is randomized and distributed according to the geometric distribution.
As a concrete application of $\listgen$, they consider the program
\[
  \listgen~(\Lam\, \_ . \listgen~(\Lam\, \_ . \Flip))
\]
which generates a list of lists of random Booleans.

Using \thereflog{}, we show that the program refines the model below.
\begin{center}
  \tikzset{every loop/.style={in=0, out=40, looseness=2}}
  \begin{tikzpicture}[node distance=1.5cm,on grid, auto]
    \node[state, accepting] (qf) {$q_{f}$};
    \node[state] (q0) [right=of qf] {$q_{0}$};
    \node[state] (q1) [right=of q0] {$q_{1}$};

    \path[->] (q0) edge [bend right] node[above] {$\frac{1}{2}$} (qf);
    \path[->] (q0) edge [bend left] node[above] {$\frac{1}{2}$} (q1);

    \path[->] (q1) edge [loop] node[right] {$\frac{1}{2}$} (q1);
    \path[->] (q1) edge [bend left] node[below] {$\frac{1}{2}$} (q0);
  \end{tikzpicture}
\end{center}
Intuitively, state $q_{0}$ corresponds to the outer application of $\listgen$ and state $q_{1}$ to the inner application.
Using the ranking super-martingale~\cite{chakarov_probabilistic_2013} $f$ that maps $q_{f}$, $q_{0}$, and $q_{1}$ to $0$, $2$, and $3$, respectively, and $\epsilon = \frac{1}{2}$ one can straightforwardly show that the model almost-surely terminates.

First, we show a specification of the inner list generator
\[
  \hoare{\spec(q_{1})}{\listgen~(\Lam\, \_ . \Flip)}{\_ \ldotp \spec(q_{0})}.
\]
The proof proceed by \ruleref{l\"{o}b} induction.
When we reach the $\Flip$ expression in $\listgen$, we apply \ruleref{rwp-coupl-rand} using the coupling $\refRcoupl{\stepdistr(q_{1})}{\unif(\bool)}{(\Lam q, b. q = \text{if}~b~\text{then}~q_{0}~\text{else}~{q_1})}$.
If $b$ is $\True$ the goal is immediate.
If $b$ is $\False$, we symbolically evaluate the $\Flip$ expression in the generator using \ruleref{rwp-rand} as this second sampling is irrelevant to termination of the program.
The induction hypothesis now finishes the proof.

The specification of the outer list generator looks as follows.
\[
  \hoare{\spec(q_{0})}{  \listgen~(\Lam\, \_ . \listgen~(\Lam\, \_ . \Flip))}{\_ \ldotp \spec(q_{f})}.
\]
The proof proceeds by \ruleref{l\"{o}b} induction and we apply \ruleref{rwp-coupl-rand} using the coupling $\refRcoupl{\stepdistr(q_{1})}{\unif(\bool)}{(\Lam q, b. q = \text{if}~b~\text{then}~q_{f}~\text{else}~{q_1})}$.
If $b$ is $\True$ the goal is immediate.
If $b$ is $\False$, the model is in state $q_{1}$ and we apply our specification for ${\listgen~(\Lam\, \_ . \Flip)}$ which returns the model in state $q_{0}$.
The induction hypothesis finishes the proof.

While our methodology is sufficient for the example at hand, we would have liked to derive a single higher-order specification of $\listgen$ that suffices for proving both of the specifications above.
However, to do so, we believe a richer notion of model is required.
The function $\listgen$ keeps invoking $f$ until it returns $\None$, \ie{}, until the corresponding stochastic process has terminated.
But when $\listgen$ is invoked in a nested fashion, the process needs to be ``restarted'' for each nested invocation.
To give a general, model-agnostic specification it seems that one would therefore need more model structure, \eg{}, recursive Markov chains \cite{DBLP:journals/jacm/EtessamiY09}.

\subsection{Lazy Real}\label{sec:lazy-real}
\newcommand{\cmpB}{\operatorname{\mathsf{cmpB}}}
\newcommand{\getB}{\operatorname{\mathsf{getB}}}
\newcommand{\cmpList}{\operatorname{\mathsf{cmpList}}}
\newcommand{\init}{\operatorname{\mathsf{init}}}
\newcommand{\cmp}{\operatorname{\mathsf{cmp}}}

\newcommand{\absolutetextcolor}[2]{%
    \textcolor{#1}{%
        \renewcommand\color[2][]{}%
    #2}%
}
\newcommand{\ghostcode}[1]{\absolutetextcolor{lightgray}{#1}}

A standard result in probability theory says that sampling a real number uniformly from the interval $[0, 1]$ is equivalent to sampling an infinite sequence of Bernoulli random variables, each independently and uniformly drawn from the set $\{0, 1\}$~\citep[Proposition~10.3.13]{cohn2013measure}.
We can think of the sequence of Bernoulli variables as representing the digits of the real number sampled from $[0, 1]$ written in binary form.
Using this representation, we can implement a procedure to sample ``exactly'' from the uniform $[0,1]$ distribution, by sampling these binary digits lazily as they are needed.
\begin{figure}
\begin{align*}
  \init \eqdef{} & \Lam \,\_ . \ghostcode{(\hspace{-0.15em}}\Alloc \None \ghostcode{, \AllocTape 1)} \\
  \cmp \eqdef{} & \Lam \loc_{1}, \ghostcode{\lbl_{1},} \loc_{2} \ghostcode{, \lbl_{2}} . \If \loc_{1} == \loc_{2} then 0 \Else \cmpList~\loc_{1}~\ghostcode{\lbl_{1}}~\loc_{2}~\ghostcode{\lbl_{2}}
\end{align*}
\begin{minipage}[t]{.45\textwidth}
  \centering
  \begin{align*}
    &\langkw{rec}~\cmpList~\loc_{1}~\ghostcode{\lbl_{1}}~\loc_{2}~\ghostcode{\lbl_{2}} = \\
    &\quad \Let (b_{1}, n_{1}) = \getB~\ghostcode{\lbl_{1}}~\loc_{1} in \\
    &\quad \Let (b_{2}, n_{2}) = \getB~\ghostcode{\lbl_{2}}~\loc_{2} in \\
    &\quad \Let c = \cmpB~b_{1}~b_{2} in \\
    &\quad \If {c == 0} then \cmpList~~n_{1}~\ghostcode{\lbl_{1}}~n_{2}~\ghostcode{\lbl_{2}} \Else c
  \end{align*}
\end{minipage}
\begin{minipage}[t]{.4\textwidth}
  \centering
  \begin{align*}
    \getB \eqdef{} \Lam \ghostcode{\lbl,}~\loc . {}
    & \MatchML \deref \loc with
      \Some~v => v
      | \None => {\begin{array}[t]{l}
                    \Let b = \Flip\ghostcode{\lbl} in \\
                    \Let n = \Alloc \None in \\
                    \Let v = (b, n) in \\
                    \loc \gets v; \\
                    v
                  \end{array}}
      end {}
  \end{align*}
\end{minipage}
\begin{align*}
  \cmpB \eqdef{} \Lam b_{1}, b_{2} . \If b_{1} < b_{2} then {-1} \Else (\If b_{2} > b_{1} then 1 \Else 0)
\end{align*}
\caption{Code for lazy uniform real sampling and comparison. Presampling annotations are shown in \textcolor{lightgray}{gray}. }
\label{fig:code-lazy-real}
\end{figure}
In this example, we consider such a lazy implementation of a sampler, along with an operation for comparing the magnitude of two lazily-sampled reals.

An implementation is shown in \cref{fig:code-lazy-real}.
The annotations shown in \textcolor{lightgray}{gray} can be ignored for now and will be discussed later.
We store the partially-sampled bits of a real number as a mutable linked list, where the head of the list is the most significant bit.
The procedure $\init$ generates a fresh random sample with no bits sampled yet, as represented by a reference initialized to ${\None}$.
Then the comparison procedure $\cmp\ l_1\ l_2$ returns $-1$ if the real represented by $l_1$ is less than $l_2$, $0$ if they are equal, and $1$ if $l_1$ is greater than $l_2$.
It is implemented by first checking whether the pointers $l_1$ and $l_2$ are equal.
If they are, it short-circuits and immediately returns $0$, since the corresponding real numbers must be the same.
Otherwise, it calls $\cmpList$, which recurses down the lists, checking bit by bit until it finds a position in the lists where the corresponding bits are not equal.
Since the bits are only sampled lazily, during a call to $\cmpList$, the next bit to be compared in one of the lists may not have been sampled yet.
To handle this, $\cmpList$ uses the wrapper function $\getB$ when accessing a bit.
The function $\getB$ either returns the bit (if it has already been sampled), and if not, it generates a fresh bit and appends it to the end of the linked list.

\emph{A priori}, comparing two lazily-sampled reals with $\cmp$ is not guaranteed to terminate, as each generated bit of the two reals could be the same, indefinitely.
However, $\cmp$ does terminate almost-surely.
We will prove this by showing a refinement with the following Markov chain model, which samples from two independent coins until they disagree, as depicted by the following diagram.
\begin{center}
  \vspace{1em}
  \begin{tikzpicture}[node distance=2cm,on grid, auto]
    \node[state, accepting] (tt) {$\top\bot$};
    \node[state, accepting] (ff) [right=of tt] {$\bot\top$};
    \node[state] (tf) [below of=tt] {$\top\top$};
    \node[state] (ft) [below of=ff] {$\bot\bot$};

    \path[->] (tf) edge [bend left] node {$\frac{1}{4}$} (tt);
    \path[->] (tf) edge node[left, xshift=-5pt] {$\frac{1}{4}$} (ff);
    \path[->] (tf) edge [bend left] node[below] {$\frac{1}{4}$} (ft);
    \path[->] (tf) edge [in=220, out=180, looseness=3] node[left] {$\frac{1}{4}$} (tf);

    \path[->] (ft) edge node[right, xshift=5pt] {$\frac{1}{4}$} (tt);
    \path[->] (ft) edge [bend right] node[right] {$\frac{1}{4}$} (ff);
    \path[->] (ft) edge [bend left] node[below] {$\frac{1}{4}$} (tf);
    \path[->] (ft) edge [in=-40, out=0, looseness=3] node[right] {$\frac{1}{4}$} (ft);
  \end{tikzpicture}
\end{center}
This model can be shown to almost-surely terminate using the ranking super-martingale $f(b_{1}, b_{2}) \eqdef \text{if}~b_{1} \neq b_{2}~\text{then}~0~\text{else}~2$ with $\epsilon = 1$.
In order to give a specification that supports multiple comparisons of (possibly) different lazily-sampled reals, we consider $N$ iterations of the model above, \ie{}, a model with state space $\bool \times \bool \times \nat$ where the last component of the state tuple represents the remaining number of times we can call $\cmp$.

As alluded to in \cref{sec:async}, the main difficulties in showing the refinement are twofold: (1) when comparing random samples where fresh bits need to be sampled on both sides, one model step corresponds to two $\Flip$ statements occurring in two different invocations of $\getB$, and (2) when comparing random samples that have already had \emph{some} bits sampled (but not the same amount), one of the samples might need to ``catch up'' by sampling additional bits first.
Presampling tapes and asynchronous coupling are key ingredients in addressing both concerns in a high-level and composable manner.
Thus, the first step in the proof is to consider a version of the program where the sampling operations are labeled with tapes.
This version of the program is obtained by including the \textcolor{lightgray}{gray} annotations shown in \cref{fig:code-lazy-real}.

To specify the lazily-sampled real operations we will make use of two predicates defined below.
\newcommand{\Lazyreal}[2]{\mathit{LazyReal}(#1, #2)}
\newcommand{\IsList}[2]{\mathit{IsList}(#1, #2)}
\newcommand{\Cmps}[1]{\mathit{Cmps}(#1)}
\begin{align*}
  \Cmps{N} &\eqdef{} \Exists b, M . \spec(b, b, M) \land M \geq N \\
  \Lazyreal{\vec{b}}{\val} &\eqdef{} \Exists \loc, \lbl, \vec{b_{1}}, \vec{b_{2}} .
  \val = (\loc, \lbl) \sep
  \vec{b} = \vec{b_{1}} \cons \vec{b_{2}} \sep
  \IsList{\loc}{\vec{b_{1}}} \sep
  \progtape{\lbl}{1}{\vec{b_{2}}}
\end{align*}
The $\Cmps{N}$ predicate keeps track of the model resource and the fact that there are \emph{at least} $N$ comparison operations left.
The $\Lazyreal{\vec{b}}{\val}$ predicate is a representation predicate that expresses that $\val$ corresponds to the lazy-sampled real denoted by the bits $\vec{b}$.
Formally, the predicate says that $\val$ is a pair of a location $\loc$ and a tape label $\lbl$ and $\vec{b}$ can be split into two sub-sequences $\vec{b_{1}}$ and $\vec{b_{2}}$ such that $\vec{b_{1}}$ corresponds to the linked list stored at location $\loc$ and $\vec{b_{2}}$ corresponds to bits that have been presampled onto the tape $\lbl$.
The $\IsList{\loc}{l}$ assertion is a standard separation logic representation predicate for linked lists~\citep{DBLP:conf/lics/Reynolds02}.

We can now give general high-level specifications to the operations.
When initializing a new lazily-sampled real, ownership of $\Lazyreal{\nil}{\val}$ is acquired for some $\val$.
\begin{align*}
  \hoare{\TRUE}{\init~\TT}{ \val \ldotp \Lazyreal{\nil}{\val}}
\end{align*}
When comparing two lazily-sampled reals, ownership of both reals are required as well as evidence that at least one comparison is left in the model.
\begin{align*}
  \hoareV
  {\Lazyreal{\vec{b_{1}}}{\val_{1}} \sep \Lazyreal{\vec{b_{2}}}{\val_{2}} \sep \Cmps{N + 1}}
  {\cmp~\val_{1}~\val_{2}}
  {\_ \ldotp \Exists \vec{b_{1}'}, \vec{b_{2}'} . \Lazyreal{\vec{b_{1}} \mdoubleplus \vec{b_{1}'}}{\val_{1}} \sep \Lazyreal{\vec{b_{2}} \mdoubleplus \vec{b_{2}'}}{\val_{2}} \sep \Cmps{N}}
\end{align*}
In the post condition we get back ownership of both reals, where more bits may have been sampled, and the number of comparisons has been decremented.
Note that the return value of $\cmp$ is not important for establishing the refinement and can hence be ignored.

While the high-level specifications are intuitive, the proof of $\cmpList$ is more intricate and goes by induction on the (pre)sampled bit sequences $\vec{b_{1}}$ and $\vec{b_{2}}$.
For the base case of the induction, which corresponds to reaching the end of both sequences, the proof uses L{\"o}b induction and presamples coupled bits to two tapes (which in turn allows us to eliminate the later modality).
If both bit sequences are non-empty, the specification follows by symbolic execution and the induction hypothesis.
If one sequence is empty and the other is not, we first sample additional bits using \ruleref{rwp-tape-empty} and continue as when both sequences are non-empty.



\subsection{Treap}\label{sec:treap}

A \emph{treap}~\citep{SeidelA96} is a randomized binary search tree structure.
Rather than using rebalancing, it relies on randomness to ensure that the tree is $O(\log n)$ height with high probability.
Searching for a key in the tree proceeds as normal in a binary search tree, but insertion makes use of randomness.
To add a new key $k$ into the tree, the insertion procedure first searches for $k$ in the tree.
If it finds $k$ is already in the tree, insertion stops and returns.
However, if $k$ is not in the tree, insertion generates a random \emph{priority} for $k$ by sampling an element $p$ independently from some totally ordered set.
How these priorities are represented and the distribution on the set they are sampled from does not matter, so long as the probability of sampling the same priority twice is low.
Once the priority $p$ is generated, insertion creates a new node containing $(k, p)$ and attaches it to the tree as a leaf node.
At this point, the priority $p$ is compared to the priority $p'$ of the node's parent $k'$.
If $p$ is greater than $p'$, then the insertion procedure performs a tree rotation, swapping the order of $(k, p)$ and $(k', p')$.
This rotation process is repeated recursively with the new parent of $k$, until $k$ either has smaller priority than all of its ancestors, or it becomes the root.

As mentioned above, it is important for the priorities of all of the nodes to be distinct.
If they are, then with high probability the tree will have $O(\log n)$ height.
In theoretical analyses of treaps~\citep{SeidelA96, DBLP:journals/jar/EberlHN20}, it is common to treat the priorities as if they are real numbers sampled from some continuous distribution, so that the probability of a collision is $0$.

\citet{EricksonTreaps} notes that one may use lazily-sampled reals, as in the previous example, to represent the priorities.
We show that with such an implementation of treaps, the insertions terminate almost surely.
Of course, this follows from the fact that the comparison operation terminates almost surely, as the previous example showed.
The motivation for this example is to demonstrate the modularity of our approach, as the treap proof does not need to know about the ``internal'' randomness and refinement proof of the lazy real comparisons.

\newcommand{\insertp}{\operatorname{\mathsf{insert}}}
\newcommand{\height}{\operatorname{\mathit{height}}}
\newcommand{\istreap}[2]{\mathit{IsTreap}(#1, #2)}

Our specification makes use of a treap representation predicate $\istreap{\val}{t}$ that expresses that $\val$ corresponds to the treap $t$ as defined below.
\begin{align*}
  \istreap{\val}{\mathsf{leaf}} \eqdef{}& \val = \None \\
  \istreap{\val}{\mathsf{node}(k, l, r)} \eqdef{}&
                                                  \Exists \loc, p, \vec{b}, \val_{l}, \val_{r} .
                                                  v = \Some~\loc \sep \progheap{\loc}{(k, p, \val_{l}, \val_{r})} \sep \\
                                        & \Lazyreal{\vec{b}}{p} \sep
                                          \istreap{\val_{l}}{l} \sep
                                          \istreap{\val_{r}}{r}
\end{align*}
If $t$ is $\mathsf{leaf}$, then $v = \None$.
If $t$ is $\mathsf{node}(k, l, r)$ then $v = \Some~\loc$ and $\loc$ is a location that contain a tuple consisting of a key $k$, a priority $p$, and two treaps $\val_{l}$ and $\val_{r}$.
We omit the usual binary search tree invariants about the ordering of keys in $l$ and $r$, since they are not needed to ensure termination.
The $\Lazyreal{\vec{b}}{p}$ assertion is the lazy real representation predicate from the previous section which states that $p$ corresponds to the lazy-sampled real denoted by the bits $\vec{b}$.
Notice how the structure of the definition closely resembles representation predicates for non-probabilistic data structures (such as lists, trees, \etc{}) and that the lazy real is logically managed abstractly through the $\Lazyreal{\vec{b}}{p}$ assertion.

Given the representation predicate, our specification of the treap $\insertp$ procedure looks as follow.
\begin{align*}
  \hoareV
  { \istreap{\val}{t} \sep \Cmps{N} \sep \height(t) \leq N}
  { \insertp~\val~k }
  { \valB .
  \istreap{\valB}{t'} \sep
  \height(t) \leq \height(t') \leq height(t) + 1 \sep
  \Cmps{M} \sep
  N - \height(t) \leq M
  }
\end{align*}
To insert a value $k$ into the treap $t$, ownership of the treap is required and evidence that \emph{at least} $\height(t)$ comparisons are left in the lazy-real model.
In return, we get ownership of a (possibly-) updated treap where the height may have been increased by one and $\height(t)$ comparisons may have been performed.
The proof proceeds in a straightforward way by applying the specification of $\cmp$ to compare the lazily-sampled priorities.

\subsection{Galton-Watson Tree}\label{sec:gwt}
\newcommand{\childdistr}{\mu}
In this example, we consider a sampler for generating Galton-Watson trees.
Galton-Watson trees are random trees generated by the following stochastic process, which proceeds through a series of rounds.
Initially, in round 1, the tree starts with a single root node, which we call generation $1$ of the tree.
In round 2, we sample a natural number $n$ from some distribution $\childdistr$, and attach $n$ children nodes to the root node.
These children are called generation $2$.
Inductively, in round $k+1$, for each node $i$ in generation $k$, we draw an independent sample $n_i$ from $\childdistr$ and attach $n_i$ children nodes to $i$.
The nodes added in round $k+1$ constitute generation $k+1$.
The process stops and is said to undergo \emph{extinction} if a generation has no nodes, \ie if all the $n_i$ in a round are $0$.

There are many algorithms for sampling Galton-Watson trees.
One approach is to essentially follow the definition above for how the trees are generated.
This can be seen as a kind of ``breadth-first'' sampling strategy, as we sample all the nodes in a given generation before moving on to any node in the next generation.
In fact, one can consider alternate strategies for ``traversing'' the tree as it is generated.
For example, \citet{DBLP:journals/siamcomp/Devroye12} describes a depth-first approach which maintains a stack containing the nodes whose children have not yet been sampled.\footnote{\citeauthor{DBLP:journals/siamcomp/Devroye12} in fact describes a variant where we want to sample a tree conditioned on the fact that it will have exactly $k$ nodes for some $k$, so as a result the algorithm aborts and restarts if $k + 1$ nodes have been generated in a tree.}

\newcommand{\run}{\mathsf{run}}
\newcommand{\qadd}{\mathsf{Stack.add}}
\newcommand{\qtake}{\mathsf{Stack.take}}
\newcommand{\qcreate}{\mathsf{Stack.create}}

\newcommand{\sampleNode}{\mathsf{sampleNode}}
\newcommand{\genTree}{\mathsf{genTree}}
\newcommand{\linit}{\mathsf{List.init}}
\newcommand{\isStack}{\operatorname{\mathit{IsStack}}}
\newcommand{\Stack}{\operatorname{\mathit{Stack}}}
\newcommand{\length}{\operatorname{\mathit{length}}}

\begin{figure}
\begin{align*}
  & \langkw{rec}~\sampleNode~d~r~s~() = \\
  & \quad \Let n = d~() in \\
  & \quad \Let f = (\Lam \,\_ . \Let r' = \Alloc~[] in (\qadd~(\sampleNode~d~r'~s)~s);~r') in \\
  & \quad r \gets \linit~n~f \\[-2.5em]
\end{align*}
\begin{minipage}[t]{0.45\textwidth}
  \begin{align*}
    \langkw{rec}~\run~s =
    \MatchML {\qtake~s} with
    \Some~f => f~(); \run~s
    | \None => ()
    end {}
  \end{align*}
\end{minipage}
\hfill
\begin{minipage}[t]{0.45\textwidth}
  \begin{align*}
    \genTree \eqdef{} \Lam d .
    & \Let r = \Alloc~[] in \\
    & \Let s = \qcreate~() in \\
    & \qadd~(\sampleNode~d~r~s)~s; \\
    & \run~s; \deref r
  \end{align*}
\end{minipage}
\caption{Implementation of a sampler for Galton-Watson trees in a higher-order event-loop style.}
\label{fig:gw-tree-code}
\end{figure}

Here, we consider an implementation of a Galton-Watson tree sampler that uses a stack to manage traversal of the tree.
However, rather than storing \emph{nodes} in the stack, we will use a higher-order implementation that stores pending \emph{tasks}, functions of type $\unittyalt \rightarrow \unittyalt$, that when invoked will carry out sampling a given node's children.
In addition, our implementation will be parameterized by a function $d$ that carries out sampling from the distribution $\childdistr$ for the number of children.

The code is shown in \cref{fig:gw-tree-code}. A node is represented by a list of pointers to its children, with an empty list representing a leaf node.
The top-level function $\genTree$ takes the child distribution sampling function $d$ as an argument.
It initializes a reference cell $r$ that contains the root node, represented as an empty list (because $r$ has no children yet) and creates an empty task stack $s$.
An initial task $\sampleNode~d~r~s$ for sampling the children of $r$ is added to the task stack.
Then, the task stack $\run$ function is called, which invokes all the tasks in the stack until there are none remaining.
When the call to $\run$ returns, $\genTree$ returns the representation of the root stored in $r$.

The actual work of carrying out sampling is done by the $\sampleNode$ task function, which takes as arguments the sampler function $d$, the node $r$ to sample for, and the task stack $s$.
This function begins by sampling the number of children $n$ from $d$.
It then uses the $\linit$ function to initialize a list of length $n$, where each element of the list is a reference to a child node generated by a call to the locally defined function $f$.
This function $f$ adds a recursive task for the child node it generates to the task stack.
The list of references is then stored back in $r$.

A key challenge here is that the task stack is re-entrant, in the sense that a task may add more tasks to the same stack it came from.
Thus, in reasoning about the recursion in $\run$, one cannot proceed by induction on the length of the stack, as the stack may grow before the recursive call.

When does such a sampler terminate?
The sampler terminates only if the tree goes extinct.
The probability of extinction is a classical problem in probability theory and depends on the distribution $\childdistr$ for the number of children.
A typical proof approach is to represent the tree generation process as a random walk Markov chain, where the position of the random walk is the number of nodes that have not yet had their children sampled.
Each transition of the walk corresponds to a node's children being sampled: if the node is in position $n+1$, it transitions to $n + k$ with probability $\mu(k)$.
(So in particular, moving from $n + 1$ to $n$ if no children are produced.)
The probability of extinction is equivalent to the probability that the walk hits $0$.
Graphically, we can represent this as follows:
\begin{center}
  \begin{tikzpicture}[node distance=1.5cm, on grid, auto,  state/.style={circle, draw, minimum size=1cm}]
    \node[state, accepting] (0) {$0$};
    \node (dots1) [right=of 0] {$\cdots$};
    \node[state] (n) [right=of dots1] {$n$};
    \node[state] (n1) [right=of n] {$n + 1$};
    \node (dots2) [right=of n1] {$\cdots$};
    \node[state] (nm) [right=of dots2] {$n + k$};

    \path[->] (n1) edge [bend left] node {$\mu(k)$} (nm);
    \path[->, dashed] (n1) edge [bend right] (dots2);
    \path[->] (n1) edge [bend left] node [below] {$\mu(0)$} (n);
  \end{tikzpicture}
\end{center}
We prove a specification that establishes a refinement between $\genTree$ and this Markov chain.

We first give general reusable specifications to the stack operations that do not concern themselves with refinement.
Since the tasks in the stack-based sampler (\ie{}, partial applications of $\sampleNode$) add elements to the stack, we need a (self-referential) specification of the stack that allows the specification of the elements in the stack to depend on the stack itself.

By working in a logic with guarded recursion, we can define \emph{guarded-recursive predicates} as a \emph{guarded fixed point} $\MU x .
t$.
Guarded fixed points have no restriction on the variance of the recursive occurrence of $x$, as long as $t$ is \emph{contractive}, \ie{}, as long as all the occurrences of $x$ in $t$ appear below a later modality.
We will use such a predicate to define a self-referential representation predicate for the stack data structure to support the higher-order nature of the Galton-Watson sampler.

We define the predicate in two steps.
First, we define an assertion $\isStack(s, n, \pred)$ that expresses that $s$ is a stack (implemented as a list) with $n$ elements, each satisfying $\pred$.
\begin{align*}
  \isStack(\pred, s, n) &\eqdef{} \Exists l . \length(l) = n \sep \IsList{s}{l} \sep \Sep_{x \in l} \pred(x)
\end{align*}
Second, we define $\Stack(\predB)(s, n)$ as a guarded fixed point.
The assertion expresses that $s$ is a stack with $n$ elements, each satisfying $\predB$, but $\predB$ may also depend on the stack assertion.
\begin{align*}
  \Stack(\predB)(s, n) &\eqdef{} \MU Q  . \isStack((\Lam \valB . \predB(Q, \valB, s)), s, n)
\end{align*}
For the fixed point to exist, the predicate $\predB : (\Val \to \nat \to \iProp) \to \Val \to \Val \to \iProp$ must be contractive in the first parameter.

With the $\Stack(\predB)$ representation predicate in hand, we prove the following generic specifications for the stack operations.
\begin{align*}
  &\hoare{\TRUE}{\qcreate}{s \ldotp \Stack(\predB)(s, 0)} \\
  &\hoare{\Stack(\predB)(s, n) \sep \predB(\Stack(\predB), \val, s)}{\qadd~\val~s}{ \_ \ldotp \Stack(\predB)(s, n + 1) } \\
  &\hoare{\Stack(\predB)(s, 0)}{\qtake~q}{\val \ldotp \val = \None \sep \Stack(\predB)(s, 0) } \\
  &\hoare{\Stack(\predB)(s, n + 1)}{\qtake~s}{\val \ldotp \Exists \valB . \val = \Some~\valB \sep \Stack(\predB)(s, n) \sep \predB(\Stack(\predB), \valB, s) }
\end{align*}
When creating a stack, ownership of an empty stack $\Stack(\predB)(s, 0)$ is obtained for a user-chosen contractive stack predicate $\predB$.
To add an element $\val$ to the stack $s$, ownership of the stack and $\predB(\Stack(\predB), \val, s)$ is required.
Finally, when popping an element $\val$ from a non-empty stack, one gets back ownership of $\predB(\Stack(\predB), \val, s)$.

\newcommand{\predgwt}{\predB_{\mathsf{GW}}}
\newcommand{\lblgwt}{\lbl_{\mathsf{GW}}}

To specify the Galton-Watson sampler, we apply our general specification of the stack operations.
Recall that $\genTree$ stores suspended tasks, \ie{} \emph{closures}, in the task stack.
Intuitively, this means that the stack predicate $\predgwt$ we instantiate the stack library with must be a \emph{specification} of the shape $\later \left(\hoare{\Stack(\predB)(s, n) \sep \ldots}{f~\TT}{\_ \ldotp \Stack(\predB)(s, n + k) \sep \ldots}\right)$ for task $f$.
The Hoare triple must be behind a later modality for $\predgwt$ to be contractive since the pre- and postcondition contains the stack representation predicate.
This may seem like an obstacle to specifying the $\run$ event-loop: when suspended tasks $f$ are taken out of the stack, the specification will still appear below a later modality which the proof has no way of eliminating.
That is, we will need to eliminate the later modality \emph{now} to apply the specification of $f$, but a coupling step---which would allow us to eliminate the later modality---only happens \emph{in the future} when $f$ is invoked.
To address this issue, the key idea is to \emph{asynchronously} couple the sampling happening in $d$ with the model during the proof of the $\run$ specification and thus eliminate the later modality earlier than otherwise permitted.
As our language and thus presampling tapes only support uniform sampling, we assume $\mu$ distributes as $\unif(N)$ for some $N$ for the remainder of the section.

Our specification of $\genTree$ looks as follows.
\begin{align*}
  &\hoare{\progtape{\lblgwt}{\tapebound}{k}}{d~\TT}{ \val \ldotp \val = k \sep \progtape{\lblgwt}{\tapebound}{\nil}} \proves \\*
  &\hoare{\spec(1) \sep \progtape{\lblgwt}{\tapebound}{\nil}}{\genTree~d}{\TRUE}
\end{align*}
The specification requires that the sampling function $d$ consumes randomness from the $\lblgwt$ tape when invoked; the $\run$ specification will populate the presampling tape through asynchronous coupling.
When creating the stack, we pick the stack predicate $\predgwt$ below.
\begin{align*}
    \predgwt(S, s, f) \eqdef{}
  \later
  \left(
  \All n, k .
  \hoare
  {S(s, n) \sep \progtape{\lblgwt}{\tapebound}{k}}
  {f~\TT}
  {\_ \ldotp S(s, n + k) \sep \progtape{\lblgwt}{\tapebound}{\nil}}
  \right)
\end{align*}
The predicate says that tasks in the stack must satisfy a specification that, given ownership of the stack and $\progtape{\lblgwt}{\tapebound}{k}$, adds $k$ new tasks to the stack.

The crux of the proof lies in the specification of $\run$.
\begin{align*}
  \hoareV
  {\Stack(\predgwt)(s, n) \sep \spec(n) \sep \progtape{\lblgwt}{\tapebound}{\nil}}
  {\run~s}
  {\_ \ldotp \Exists m . \Stack(\predgwt)(s, m) \sep \spec(m) \sep \progtape{\lblgwt}{\tapebound}{\nil} }
\end{align*}
The specification of $\run$ requires ownership of a stack $s$ with $n$ elements, the model resource at state $n$, and the $\lblgwt$ presampling tape.
The proof also proceeds by L{\"o}b induction.
To eliminate \emph{both} the later modality in front of the induction hypothesis \emph{and} the later modality in front of the specification of the task retrieved from the queue, the proof asynchronously couples the model transition from $n$ with a presampling onto $\lblgwt$.


%% file: related-work.tex
\section{Related Work}

\paragraph{AST of first-order programs}

The study of termination of stochastic processes has a long history in probability theory, and there are broad classes of techniques for proving termination in the theory of Markov chains and branching processes. Many of these techniques have been adapted to formal methods for proving termination of probabilistic programs.

For example, \citet{chakarov_probabilistic_2013} use ranking super-martingales (RSM) for proving almost-sure termination.
They implement an analysis that can automatically discover the existence of RSMs.
A number of follow-up works have extended the scope and applicability of ranking super-martingales~\citep{DBLP:conf/popl/FioritiH15, fu_termination_2019, McIverMKK18, DBLP:journals/pacmpl/AgrawalC018}. 
Several of these works develop program logics for first-order imperative probabilistic programs, in which the primitive mechanism for looping is a {\tt while} loop construct.
Ranking super-martingales are used in the conditions for the {\tt while} loop, analogous to the way ranking functions are used in standard, non-probabilistic Hoare logic's variant rule for loops.
In principle, similar rules could be devised for particular schemes and patterns of recursion in a higher-order language like \thelang{}, but it would be challenging to devise general purpose rules, that would apply to examples like the treap or the Galton-Watson sampler from \cref{sec:examples}.

\citet{DBLP:conf/fossacs/AronsPZ03} give an approach for reducing almost-sure termination of a randomized program to may-termination under a non-deterministic semantics.
For this non-deterministic semantics, their \emph{planner} rule allows one to pre-select a finite sequence of outcomes for random choices.
If the program can be shown to terminate when this finite sequence of outcomes occurs, then it must terminate almost-surely under the standard probabilistic semantics.
\citet{DBLP:conf/cav/EsparzaGK12} generalize this rule to more flexible forms of finite sequences, which they call \emph{patterns}.
This generalization is complete for \emph{weakly finite} programs, which from each starting state can only reach a finite number of states.
An advantage of this approach is that it allows one to then apply tools and algorithms for automatically analyzing non-probabilistic program termination.
However, examples like the unbounded random walk or the Galton-Watson tree are not weakly finite in this sense, and appear to be beyond the scope of these techniques.
The soundness of these approaches is justified by the Borel-Cantelli lemma, which shows that the pre-selected sequence of outcomes must occur infinitely often in any non-terminating execution.
The Borel-Cantelli lemma is an example of a zero-one law, and similar zero-one laws have been used to justify other proof rules for almost sure termination~\citep{DBLP:journals/toplas/HartSP83}.

\emph{Positive} almost-sure termination (PAST) is a stronger property \cite{DBLP:conf/rta/BournezG05, DBLP:journals/pacmpl/MajumdarS24}, stating that the expected number of total transitions for a program is finite.
A number of methods have been developed for proving bounds on the expected running time of a program, thereby implying almost-sure termination. 
For example, \citet{kaminski_weakest_2016} develop a weakest-precondition style calculus for bounding expected running time.
\citet{DBLP:conf/pldi/NgoC018} develop an extension of automatic amortized resource analysis (AARA) for bounding expectations of resource use in randomized programs.
Some of the examples considered in \cref{sec:examples}, such as the random walk, are AST but have infinite expected running time, so techniques based on bounding expected running time cannot be used to prove that they are AST.

Although our focus has been on using \thereflog{} to prove AST, the lower bound on termination probabilities established by \cref{thm:soundness} applies even when the model does not terminate with probability 1.
Thus, \thereflog{} can be used to show lower-bounds on termination probabilities for programs that are not AST.
\citet{DBLP:journals/pacmpl/FengCSKKZ23} develop a weakest pre-expectation calculus for proving lower bounds on expected values of quantities in non-AST programs.

\paragraph{AST of higher-order programs}

\citet{DBLP:conf/tphol/Hurd02} develops theory for proving termination of a monadic embedding of randomized programs in HOL.
He defines a {\tt while} combinator for this embedding and proves analogues of 0-1 rules for almost-sure termination by \citet{DBLP:journals/toplas/HartSP83}.

\citet{DBLP:journals/pacmpl/AvanziniBL21} present a technique to reason about the expected runtime of programs written in a probabilistic lambda calculus.
They use a continuation-passing style translation, where the continuation maps inputs to expected runtimes, thus turning the program into a cost transformer.

Several works introduce type systems that imply termination or expected bounds on higher-order programs. \citet{DBLP:conf/esop/LagoG17} present a probabilistic variant of sized types that ensures almost-sure termination of well-typed programs.
In this system, a function body's type is associated with a particular kind of random walk called a sized random walk.
The typing rule for {\tt letrec} has a premise that requires the corresponding random walk to be almost-surely terminating.
AST for these sized random walks is shown to be decidable.
\citet{DBLP:journals/pacmpl/LagoFR21} develop an intersection type systems for capturing both almost-sure and positive almost-sure termination.
\citet{DBLP:journals/pacmpl/WangKH20} present a probabilistic variant of RaML, a higher-order language with a type system that does amortized resource analysis to produce bounds on expected resource consumption automatically.

\citet{DBLP:journals/lmcs/KobayashiLG20} introduce  probabilistic higher-order recursion schemes (PHORS), a probabilistic variant of the HORS considered in higher-order model checking.
They show that the decision problem for almost-sure termination of order-2 PHORS is undecidable, and introduce a sound but incomplete procedure for bounding termination probabilities of order-2 PHORS.
An implementation of their procedure is applied to several PHORS, including one that is equivalent to the listgen example in \cref{sec:examples}.
In place of Markov chains, PHORS could be used as the models in \thereflog{}, in order to prove that a particular PHORS is an adequate abstraction of a program.

In contrast to \thereflog{}, the above works deal with higher-order calculi \emph{without} imperative state. 

\paragraph{Guarded recursion for termination and termination-preserving refinement}
We have already discussed closely related uses of guarded recursion for termination-preserving refinement in \cref{sec:sound-compare}.
Other works have made use of guarded recursion or step-indexed models to prove termination or termination-preserving refinement.
\citet{DBLP:journals/pacmpl/SpiesKD21} construct a transfinite step-indexed logical relation to show that a linear type system implies termination.
SeLoC~\citep{DBLP:conf/sp/FruminKB21} is a relational logic for proving termination-sensitive noninterference properties of higher-order programs.
To achieve the intended security property, the simulation relation that their program logic encodes is much stricter than the one we have considered here, so that later modalities are only eliminated in rules where \emph{both} the program and its model take a simultaneous step.
In contrast, \citet{tiniris} develop a model of termination-\emph{in}sensitive noninterference where later modalities are eliminated when \emph{either} the program or the model take a step.

Much as PAST implies AST, establishing an upper bound on the number of steps a program takes also implies termination.
Several works have proved resource bounds in step-indexed program logics~\citep{tcred-iris, thunks} using the technique of time credits~\citep{atkey, unionfind}.
In this approach, a later modality is eliminated on every program step, which might at first appear to allow infinite looping by L\"{o}b induction.
However, this is ruled out by requiring a finite resource called a time credit to be spent on steps of execution.
Such an approach could be adapted to deriving expected bounds on probabilistic programs, but this would not be applicable for proving AST of the symmetric random walk, as it is not PAST.

\paragraph{Guarded recursion for probabilistic refinement}
Guarded recursion has also been used to define logical relations models and logics for refinement of probabilistic programs.
\citet{bizjak-birkedal} construct such a logical relation for a language similar to our \thelang{}, but without higher-order state.
\citet{aguirre-birkedal} extend this language with support for countable non-deterministic choice.
In the presence of non-determinism, a program has a range of probabilities of termination based on how non-determinism is resolved by a scheduler.
They develop a logical relation for proving equivalences with respect to the may and must-termination probabilities (the maximal and minimal probabilities across all schedulers).
\citet{DBLP:journals/pacmpl/WandCGC18} build a step-indexed logical relation for proving contextual equivalences for a higher-order language with sampling from continuous distributions.
Finally, Clutch~\citep{clutch} is a program logic based on Iris for proving contextual refinements (that do not preserve termination of programs) written in a language like \thelang{}.
The general structure of our coupling precondition is inspired by the coupling modality of Clutch.
However, our definitions, method for preserving termination, and the soundness proof is new and entirely different.
We re-use the idea of pre-sampling tapes from Clutch, which is mostly orthogonal to the question of termination preservation. 

In the above works, the defined refinement relations lead to lower bounds that are effectively in the \emph{opposite} direction of \cref{thm:soundness}.
In other words, translating their approaches to our setting leads to a soundness theorem in which the termination probability of the \emph{program} is a lower bound on the termination probability of the \emph{model}.
This is because their approaches effectively allow for later elimination when the program takes a step, as opposed to when the model takes a step as in \thereflog{}.\footnote{The first three of these works use explicitly step-indexed models, as opposed to the ``logical'' approach to step indexing~\citep{DBLP:journals/corr/abs-1103-0510} with $\later$ modalities. Nevertheless, they decrement the step index when the program on the left side of their refinement relation takes a step, which corresponds to eliminating a later modality as we have described.}
Since the goal in those works is to prove \emph{equivalences} by proving refinements in both directions, the direction of this inequality is adequate for their purposes.

\citet{DBLP:conf/esop/0001BBBG018} present a logic to prove couplings between
pairs of infinite runs over Markov chains defined in a probabilistic guarded
lambda calculus, without considering any form of state. Markov chains are
defined as distributions over infinite streams. Productivity (the fact that a
step will eventually be taken) is dual
to termination, and is ensured by their type system.

Polaris~\citep{polaris} is a concurrent program logic based on Iris for proving a coupling between a randomized program and a more abstract model.
The soundness theorem for Polaris allows bounds on probabilities and expectations in the model to be translated into bounds on the program across schedulers.
However, these bounds only apply under schedulers for which the program terminates in a bounded number of steps.
Thus, Polaris is a kind of partial correctness logic, as its soundness theorem \emph{assumes} a property that is already stronger than AST.


%% file: conclusion.tex
\section{Conclusion}

We have presented \thereflog{}, a logic for proving termination-preserving refinements between higher-order probabilistic programs and more abstract models.
For proving such refinements, \thereflog{} combines powerful techniques such as L\"{o}b induction and couplings.
We have demonstrated \thereflog{} on several examples, including ones that are outside the scope of prior methods and approaches for proving almost-sure termination of randomized programs.
A natural future direction would be to extend \thereflog{} with support for adversarial non-determinism in programs and models, and to consider classes of abstract models with more structure, such as PHORS and recursive Markov chains, with the aim of developing more composable specifications.


%% file: appendix.tex
\section{Refinement Logic with Invariant Mask Annotations}\label{app:logic}
\small
\begin{lemma}[Program logic rules]
  \begin{align*}
    \expr_{1} \purestep \expr_{2} \ast \rwpre{\expr_{2}}[\mask]{\pred} &\proves \rwpre{\expr_{1}}[\mask]{\pred} \ruletag{rwp-pure}\\
    \All \loc . \progheap{\loc}{\val} \wand \pred(\loc) &\proves \rwpre{\Alloc \val}[\mask]{\pred} \ruletag{rwp-alloc} \\
    (\progheap{\loc}{\val} \wand \pred(\val))\sep {\progheap{\loc}{\val}} &\proves \rwpre{\deref\loc}[\mask]{\pred} \ruletag{rwp-load} \\
    (\progheap{\loc}{\valB} \wand \pred(\TT))\sep {\progheap{\loc}{\val}} &\proves \rwpre{\loc \gets \valB}[\mask]{\pred} \ruletag{rwp-store} \\
    \All n \leq \tapebound . \pred(n) &\proves \rwpre{\Rand \tapebound}[\mask]{\pred} \ruletag{rwp-rand} \\
    \pred(\val) &\proves \rwpre{\val}[\mask]{\pred} \ruletag{rwp-val} \\
    \rwpre{\expr}[\mask]{\val . \rwpre{\fillctx\lctx[\val]}[\mask]{\pred}} &\proves \rwpre{\fillctx\lctx[\expr]}[\mask]{\pred} \ruletag{rwp-bind} \\
    (\All \val . \Psi(\val) \wand \pred(\val)) \sep \rwpre{\expr}[\mask]{\Psi} &\proves \rwpre{\expr}[\mask]{\pred} \ruletag{rwp-mono} \\
    \prop \sep \rwpre{\expr}[\mask]{\pred} &\proves \rwpre{\expr}[\mask]{v . \, \prop \sep \pred(v)} \ruletag{rwp-frame} \\
    \All \lbl . \progtape{\lbl}{\tapebound}{\nil} \wand \pred(\lbl) & \proves \rwpre{\AllocTape \tapebound}[\mask]{\pred}\ruletag{rwp-tape-alloc} \\
    (\progtape{\lbl}{\tapebound}{\tape} \wand \pred(n)) \sep \progtape{\lbl}{\tapebound}{n \cons \tape} &\proves \rwpre{\Rand \tapebound~\lbl}[\mask]{\pred} \ruletag{rwp-tape-empty} \\
    (\All n \leq \tapebound . \progtape{\lbl}{\tapebound}{\nil} \wand \pred(n)) \sep \progtape{\lbl}{\tapebound}{\nil} &\proves \rwpre{\Rand \tapebound~\lbl}[\mask]{\pred} \ruletag{rwp-tape}
  \end{align*}
\end{lemma}

\begin{lemma}[Model rules]
  \normalfont
  \begin{mathpar}
    \inferH{rwp-coupl-rand}
    { \refRcoupl{\stepdistr(\mstate_1)}{\unif(\tapebound)}{R} \\
      \vdash \All (\mstate_2, n) \in R . (\spec(\mstate_2) \ast P) \wand \rwpre{n}[\mask]{\pred}
    }
    {\spec(\mstate_1) \ast \later P \vdash \rwpre{\Rand \tapebound}[\mask]{ \pred }}
    \and
    \inferH{rwp-spec-det}
    {
      \stepdistr(\mstate_1)(\mstate_2) = 1 \\
      \spec(\mstate_2) \ast \prop \vdash \rwpre{\expr}[\mask]{\pred}
    }
    {\spec(\mstate_1) \ast \later \prop \vdash \rwpre{\expr}[\mask]{\pred}}
  \end{mathpar}
\end{lemma}

\section{Semantic Model}
\small
\begin{definition}[Refinement weakest precondition]
  \begin{align*}
    \rwpre{\expr_{1}}[\mask]{\pred} \eqdef{} \lfp W .
    & (\expr_{1} \in \Val \sep \pvs[\mask] \pred(\expr_{1})) \lor{} \\
    & (\expr_{1} \not\in \Val \sep \All \state_{1}, \mstate_{1} . \specinterp(\mstate_{1}) \ast \stateinterp(\state_{1}) \wand{} \pvs[\mask][\emptyset] \\
    & \quad \execCoupl{\mstate_{1}}{(\expr_{1}, \state_{1})}{\mstate_{2}, (\expr_{2}, \state_{2}) .
      \pvs[\emptyset][\mask] \specinterp(\mstate_{2}) \sep \stateinterp(\state_{2}) \sep W(\expr_{2}, \mask, \pred)}) \\
  \end{align*}
\end{definition}

\begin{definition}[Coupling modality]
  \begin{mathpar}
    \infer
    { \red(\cfg_1) \\
      \refRcoupl{\mret(\mstate)}{\stepdistr(\cfg_{1})}{R} \\
      \All s \in R . \pvs[\emptyset] \predB(s)
    }
    {\execCoupl{\mstate}{\cfg_{1}}{\predB}}
    \and
    \infer
    { \red(\cfg_1) \\
      \red(\mstate_1) \\
      \refRcoupl{\stepdistr(\mstate_1)}{\stepdistr(\cfg_1)}{R} \\
      \All s \in R . \later \pvs[\emptyset] \predB(s)
    }
    {\execCoupl{\mstate_1}{\cfg_1}{\predB}}
    \and
    \infer
    { \red(\mstate_1) \\
      \refRcoupl{\stepdistr(\mstate_{1})}{\mret(\cfg)}{R} \\
      \All (\mstate_{2}, \cfg) \in R . \later \pvs[\emptyset] \execCoupl{\mstate_2}{\cfg}{\predB}
    }
    {\execCoupl{\mstate_1}{\cfg}{\predB}}
    \and
    \infer
    { \red(\mstate_1) \\
      l \subseteq \dom(\state_1) \\
      \refRcoupl{\stepdistr(\mstate_{1})}{\foldM(\statestepdistr, \state_{1}, l)}{R} \\
      \All (\mstate_{2}, \state_{2}) \in R . \later \pvs[\emptyset] \execCoupl{\mstate_2}{(\expr_1, \state_2)}{\predB}
    }
    {\execCoupl{(\expr_1, \state_1)}{\mstate_1}{\predB}}
  \end{mathpar}
\end{definition}
